\pdfoutput=1
\documentclass[oribibl]{llncs}
\usepackage[utf8]{inputenc}

\newif\ifarxiv
\newcommand{\arxiv}[2]{\ifarxiv #1 \else #2 \fi}
\arxivtrue

\usepackage{amsfonts}
\usepackage{amsmath}
\usepackage{mathrsfs}
\usepackage{amssymb}
\usepackage{dsfont}
\newcommand{\faBan}{\ensuremath{\oslash}}

\usepackage{microtype}

\usepackage{algorithm}
\usepackage[noend]{algpseudocode}

\usepackage{etoolbox}
\makeatletter
\patchcmd{\ALG@doentity}{\item[]\nointerlineskip}{}{}{}
\makeatother

\makeatletter
\renewcommand\fnum@algorithm{\fname@algorithm~\thealgorithm.}
\makeatother

\usepackage{needspace}
\makeatletter
\newenvironment{breakablealgorithm}
  {
 	\needspace{3\baselineskip} 
   \begin{center}
     \refstepcounter{algorithm}
     \hrule height.8pt depth0pt \kern2pt
     \renewcommand{\caption}[2][\relax]{
       {\raggedright\textbf{\ALG@name~\thealgorithm.} ##2\par}%
       \ifx\relax##1\relax 
         \addcontentsline{loa}{algorithm}{\protect\numberline{\thealgorithm}##2}%
       \else 
         \addcontentsline{loa}{algorithm}{\protect\numberline{\thealgorithm}##1}%
       \fi
       \kern2pt\hrule\kern2pt
     }
		\small
  }{
     \kern2pt\hrule\relax
   \end{center}
  }
\makeatother

\renewenvironment{algorithm}{\begin{breakablealgorithm}}{\end{breakablealgorithm}\ignorespacesafterend}

\newcommand{\ket}[1]{\vert{#1}\rangle}
\newcommand{\bra}[1]{\langle{#1}\vert}
\newcommand{\id}{\mathbb{I}}    

\newcommand{\Encode}{\mathsf{Encode}}
\newcommand{\Decode}{\mathsf{Decode}}
\newcommand{\Sign}{\mathsf{Sign}}
\newcommand{\Tag}{\mathsf{Tag}}
\newcommand{\Enc}{\mathsf{Enc}}

\newcommand{\Classical}{\mathsf{Classical}}
\newcommand{\Measure}{\mathsf{Measure}}
\newcommand{\Conditional}{\mathsf{Cond}}
\newcommand{\HE}{\mathsf{HE}}
\newcommand{\TTP}{\mathsf{TrapTP}}
\newcommand{\TC}{\mathsf{TC}}

\newcommand{\Eval}{\mathsf{Eval}}
\newcommand{\Ver}{\mathsf{Ver}}
\newcommand{\Dec}{\mathsf{Dec}}
\newcommand{\VerDec}{\mathsf{VerDec}}
\newcommand{\VerDecQubit}{\mathsf{VerDecQubit}}
\newcommand{\VerDecMeasurement}{\mathsf{VerDecMeasurement}}
\newcommand{\KeyGen}{\mathsf{KeyGen}}
\newcommand{\Post}{\mathsf{Post}}
\newcommand{\Fake}{\mathsf{Fake}}
\newcommand{\GadgetGen}{\mathsf{GadgetGen}}

\newcommand{\T}{\mathsf{T}}
\renewcommand{\P}{\mathsf{P}}
\renewcommand{\H}{\mathsf{H}}
\newcommand{\CNOT}{\mathsf{CNOT}}
\newcommand{\CSS}{\mathsf{CSS}}
\newcommand{\MAC}{\mathsf{MAC}}
\newcommand{\gin}{\gamma^\mathsf{in}}
\newcommand{\gmid}{\gamma^\mathsf{mid}}
\newcommand{\gout}{\gamma^\mathsf{out}}
\newcommand{\flag}{\mathit{flag}}
\newcommand{\CheckLog}{\mathsf{CheckLog}}

\newcommand{\barctrl}{\mathsf{\overline{ctrl}\mbox{-}}}
\newcommand{\Real}{\mathsf{Real}}
\newcommand{\Ideal}{\mathsf{Ideal}}

\newcommand{\advA}{\mathcal{A}}

\newcommand{\simS}{\mathcal{S}}

\newcommand{\tr}{\mbox{tr}}
\DeclareMathOperator*{\Ex}{\mathbb{E}}

\newcommand{\disD}{\mathcal{D}}
\newcommand{\inM}{\mathcal{M}}

\newcommand{\VerGame}[1]{\mathsf{VerGame}_{\advA,#1}(\kappa)}
\newcommand{\VerGamePrime}[1]{\mathsf{VerGame}_{\advA',#1}(\kappa)}
\newcommand{\VerTwoGame}[1]{\mathsf{VerGame}^2_{\advA,#1}(\kappa)}
\newcommand{\VerTwoGamePrime}[1]{\mathsf{VerGame}^2_{\advA',#1}(\kappa)}
\newcommand{\HybridGame}[1]{\mathsf{Hyb}_{\advA,#1}(\kappa)}
\newcommand{\AdvHyb}[2]{\mathsf{AdvHyb}_{#1}^{#2}(\advA, \kappa)}

\renewcommand{\tr}{\text{\textnormal{tr}}}

\usepackage{tikz}
\usetikzlibrary{decorations.pathmorphing,decorations.pathreplacing}
\usepackage{todonotes}
\renewcommand{\todo}[1]{}
\renewcommand{\todo}[2][]{}
\usepackage{fullpage}

\renewcommand{\S}{\ensuremath{\mathcal{S}}}

\newcommand{\X}{\mathsf{X}}
\newcommand{\Y}{\mathcal{Y}}
\newcommand{\Z}{\mathsf{Z}}

\newcommand{\accstate}{\egoketbra{\text{\textnormal{acc}}}}
\newcommand{\rejstate}{\egoketbra{\text{\textnormal{rej}}}}
\newcommand{\accstatepure}{\ket{\text{acc}}}
\newcommand{\rejstatepure}{\ket{\text{rej}}}

\newcommand{\ketbra}[2]{\ket{#1}\bra{#2}}			
\newcommand{\egoketbra}[1]{\ketbra{#1}{#1}}			
	
\newcommand{\meas}{
	\begin{tikzpicture}
	\filldraw[fill=white] (0,.25) rectangle (.7,-.25);
	\draw (.67,-.1) arc (50:130:.5);
	\draw (.35,-.2)--(.525,.2);
	\end{tikzpicture}
}

\newcommand{\remove}{\text{\faBan}}
\newcommand{\enc}[1]{\widetilde{#1}}

\usepackage[numbers,sectionbib]{natbib}
\usepackage[colorlinks=true,urlcolor=webblue,linkcolor=webgreen,filecolor=webblue,citecolor=webgreen,pdfpagemode=UseOutlines,pdfstartview=FitH,pdfpagelayout=OneColumn,bookmarks=true]{hyperref}
\definecolor{webgreen}{rgb}{0,.5,0}
\definecolor{webblue}{rgb}{0,0,.5}

\newcommand{\poly}{\operatorname{poly}}
\newcommand{\negl}{\operatorname{negl}}
\newcommand{\algo}{\mathcal}
\newcommand{\cacc}{\ensuremath{\mathsf{acc}}}
\newcommand{\crej}{\ensuremath{\mathsf{rej}}}
\newcommand{\evk}{\ensuremath{\mathsf{evk}}}
\newcommand{\from}{\leftarrow}
\newcommand{\NC}{\textbf{NC}}
\newcommand{\states}{\mathfrak D}
\newcommand{\hi}{\mathcal H}
\newcommand{\OTP}{\ensuremath{\mathsf{OTP}}}
\newcommand{\create}{\ensuremath{\mathsf{create}}}

\newtheorem{game}{Game}

\title{Quantum Fully Homomorphic Encryption With Verification}
\author{}\institute{}
\author{Gorjan Alagic\inst{1,2} \and Yfke Dulek\inst{3} \and Christian Schaffner\inst{3} \and Florian Speelman\inst{4}}
\institute{Joint Center for Quantum Information and Computer Science, University of Maryland, College Park, MD
\and
National Institute of Standards and Technology, Gaithersburg, MD
\and
CWI, QuSoft, and University of Amsterdam
\and
QMATH, Department of Mathematical Sciences, University of Copenhagen}

\begin{document}
\maketitle
\begin{abstract}
Fully-homomorphic encryption (FHE) enables computation on encrypted data while maintaining secrecy. Recent research has shown that such schemes exist even for quantum computation. Given the numerous applications of classical FHE (zero-knowledge proofs, secure two-party computation, obfuscation, etc.) it is reasonable to hope that quantum FHE (or QFHE) will lead to many new results in the quantum setting. However, a crucial ingredient in almost all applications of FHE is \emph{circuit verification}. Classically, verification is performed by checking a transcript of the homomorphic computation. Quantumly, this strategy is impossible due to no-cloning. This leads to an important open question: can quantum computations be delegated and verified in a non-interactive manner?

In this work, we answer this question in the affirmative, by constructing a scheme for QFHE with verification (vQFHE). Our scheme provides authenticated encryption, and enables arbitrary polynomial-time quantum computations without the need of interaction between client and server. Verification is almost entirely classical; for computations that start and end with classical states, it is completely classical. As a first application, we show how to construct quantum one-time programs from classical one-time programs and vQFHE.

\end{abstract}

\section{Introduction}

The 2009 discovery of fully-homomorphic encryption (FHE) in classical cryptography is widely considered to be one of the major breakthroughs of the field. Unlike standard encryption, FHE enables non-interactive computation on encrypted data even by parties that do not hold the decryption key. Crucially, the input, output, and all intermediate states of the computation remain encrypted, and thus hidden from the computing party. While FHE has some obvious applications (e.g., cloud computing), its importance in cryptography stems from its wide-ranging applications to other cryptographic scenarios. For instance, FHE can be used to construct secure two-party computation, efficient zero-knowledge proofs for NP, and indistinguishability obfuscation~\cite{blogBB12, GGHRSW13}. In fact, the breadth of its usefulness has led some to dub FHE ``the swiss army knife of cryptography''\cite{blogBB12}. 

Recent progress on constructing quantum computers has led to theoretical research on ``cloud-based'' quantum computing. In such a setting, it is natural to ask whether users can keep their data secret from the server that performs the quantum computation.  A recently-constructed quantum fully-homomorphic encryption (QFHE) scheme shows that this can be done in a single round of interaction~\cite{DSS16}. This discovery raises an important question: do the numerous classical applications of FHE have suitable quantum analogues? As it turns out, most of the classical applications require an additional property which is simple classically, but non-trivial quantumly. That property is \emph{verification}: the ability of the user to check that the final ciphertext produced by the server is indeed the result of a particular computation, homomorphically applied to the initial user-generated ciphertext. In the classical case, this is a simple matter: the server makes a copy of each intermediate computation step, and provides the user with all these copies. In the quantum case, such a ``transcript" would appear to violate no-cloning. the user simply checks a transcript generated by the server. In the quantum case, this would violate no-cloning. In fact, one might suspect that the no-cloning theorem prevents non-interactive quantum verification \emph{in principle}.

In this work, we show that verification of homomorphic quantum computations is in fact possible. We construct a new QFHE scheme which allows the server to generate a ``computation log'' which can certify to the user that a particular homomorphic quantum computation was performed on the ciphertext. The computation log itself is purely classical, and most (in some cases, all) of the verification can be performed on a classical computer. Unlike in all previous quantum homomorphic schemes, the underlying encryption is now authenticated. 

Verification immediately yields new applications of QFHE, e.g., allowing users of a ``quantum cloud service'' to certify the server's computations. Verified QFHE (or vQFHE) also leads to a simple construction of quantum one-time programs (qOTPs)~\cite{BGS13}. In this construction, the qOTP for a functionality $\Phi$ consists of an evaluation key and a classical OTP which performs vQFHE verification for $\Phi$ only. Finding other applications of vQFHE (including appropriate analogues of all classical applications) is the subject of ongoing work.

\subsubsection{Related Work.}

Classical FHE was first constructed by Gentry in 2009~\cite{Gentry09}. For us, the scheme of Brakerski and Vaikuntanathan~\cite{BV11} is of note: it has decryption in $\NC^1$ and is believed to be quantum-secure. Quantumly, partially-homomorphic (or partially-compact) schemes were proposed by Broadbent and Jeffery~\cite{BJ15}. The first fully-homomorphic (leveled) scheme was constructed by Dulek, Schaffner and Speelman~\cite{DSS16}. Recently, Mahadev proposed a scheme, based on classical indistinguishability obfuscation, in which the user is completely classical~\cite{Mah17}. A parallel line of work has attempted to produce QFHE with information-theoretic security~\cite{FPY14, FTO15, CFKOT16, NS17}. There has also been significant research on delegating quantum computation interactively (see, e.g.,~\cite{ABE08, BFK09, CGJV17}). Another notable interactive approach is quantum computation on authenticated data (QCAD), which was used to construct quantum one-time programs from classical one-time programs~\cite{BGS13} and zero-knowledge proofs for QMA~\cite{BJSW16}.

\subsubsection{Summary of Results.}

Our results concern a new primitive: verified QFHE. A standard QFHE scheme consists of four algorithms: $\KeyGen$, $\Enc$, $\Eval$ and $\Dec$~\cite{BJ15, DSS16}. We define vQFHE similarly, with two changes: (i.) $\Eval$ provides an extra classical ``computation log'' output; (ii.) decryption is now called $\VerDec$, and accepts a ciphertext, a circuit description $C$, and a computation log. Informally, correctness then demands that, for all keys $k$ and circuits $C$ acting on plaintexts,
\begin{equation}\label{eq:corr-informal}
\VerDec^C_k \circ \Eval^C_\evk \circ \Enc_k = \Phi_C\,.
\end{equation}
A crucial parameter is the relative difficulty of performing $C$ and $\VerDec^C_k$. In a nontrivial scheme, the latter must be simpler. In our case, $C$ is an arbitrary poly-size quantum circuit and $\VerDec^C_k$ is almost entirely classical.

\paragraph{Security of verified QFHE.} 

Informally, security should require that, if a server deviates significantly from the map $\Eval^C_k$ in \eqref{eq:corr-informal}, then $\VerDec^C_k$ will reject. 
\begin{enumerate}
\item \textbf{Semantic security (SEM-VER)}. Consider a QPT adversary $\algo A$ which manipulates a ciphertext (and side info) and declares a circuit, as in Figure~\ref{fig:SEM-VER} (top). This defines a channel $\Phi_{\algo A} := \VerDec \circ \algo A \circ \Enc$. A simulator $\algo S$ does not receive or output a ciphertext, but does declare a circuit; this defines a channel $\Phi_{\algo S}$ which first runs $\algo S$ and then runs a circuit on the plaintext based on the outputs of $\algo S$.  We say that a vQFHE scheme is semantically secure (SEM-VER) if for all adversaries $\algo A$ there exists a simulator $\algo S$ such that the channels $\Phi_{\algo A}$ and $\Phi_{\algo S}$ are computationally indistinguishable. 

\item \textbf{Indistinguishability (IND-VER)}. Consider the following security game. Based on a hidden coin flip $b$, $\algo A$ participates in one of two protocols. For $b = 0$, this is normal vQFHE. For $b=1$, this is a modified execution, where we secretly swap out the plaintext $\rho_{\algo A}$ to a private register (replacing it with a fixed state), apply the desired circuit to $\rho_{\algo A}$, and then swap $\rho_{\algo A}$ back in; we then discard this plaintext if $\VerDec$ rejects the outputs of $\algo A$. Upon receiving the final plaintext of the protocol, $\algo A$ must guess the bit $b$. A vQFHE scheme is IND-VER if, for all $\algo A$, the success probability is at most $1/2 + \negl(n)$.

\item \textbf{New relations between security definitions}. If we restrict SEM-VER to empty circuit case, we recover (the computational version of) the definition of quantum authentication~\cite{DNS12, BW16}. SEM-VER (resp., IND-VER) generalizes computational semantic security SEM (resp., indistinguishability IND) for quantum encryption~\cite{BJ15, ABFGSS16}. We generalize SEM $\Leftrightarrow$ IND~\cite{ABFGSS16} as follows.

\begin{theorem}
A vQFHE scheme satisfies SEM-VER iff it satisfies IND-VER.
\end{theorem}

\end{enumerate}

\paragraph{A scheme for vQFHE for poly-size quantum circuits.}

Our main result is a vQFHE scheme which admits verification of arbitrary polynomial-size quantum circuits. The verification in our scheme is almost entirely classical. In fact, we can verify classical input/output computations using purely classical verification. The main technical ingredients are (i.) classical FHE with $\NC^1$ decryption~\cite{BV11}, (ii.) the trap code for computing on authenticated quantum data~\cite{SP00, BGS13, BW16}, and (iii.) the ``garden-hose gadgets'' from the first QFHE scheme~\cite{DSS16}. The scheme is called $\TTP$; a brief sketch is as follows.
\begin{enumerate}
\item \textbf{Key Generation ($\KeyGen$).} We generate keys for the classical FHE scheme, as well as some encrypted auxiliary states (see evaluation below). This procedure requires the generation of single-qubit and two-qubit states from a small fixed set, performing Bell measurements and Pauli gates, and executing the encoding procedure of a quantum error-correcting code on which the trap code is based.
\item \textbf{Encryption ($\Enc$).} We encrypt each qubit of the plaintext using the trap code, and encrypt the trap code keys using the FHE scheme. This again requires the ability to perform Paulis, execute an error-correcting encoding, and the generation of basic single-qubit states.
\item \textbf{Evaluation ($\Eval$).} Paulis and $\CNOT$ are evaluated as in the trap code; keys are updated via FHE evaluation. To measure a qubit, we measure all ciphertext qubits and place the outcomes in the log. To apply $\P$ or $\H$, we use encrypted magic states (from the eval key) plus the aforementioned gates. Applying $\T$ requires a magic state and an encrypted ``garden-hose gadget'' (because the $\T$-gate magic state circuit applies a $\P$-gate conditioned on a measurement outcome). In addition to all of the measurement outcomes, the log also contains a transcript of all the classical FHE computations.
\item \textbf{Verified decryption ($\VerDec$).} We check the correctness and consistency of the classical FHE transcript, the measurement outcomes, and the claimed circuit. The result of this computation is a set of keys for the trap code, which are correct provided that $\Eval$ was performed honestly. We decrypt using these keys and output either a plaintext or \textsf{reject}. In terms of quantum capabilities, decryption requires executing the decoding procedure of the error-correcting code, computational-basis and Hadamard-basis measurements, and Paulis.
\end{enumerate}
Our scheme is \emph{compact}: the number of elementary quantum operations performed by $\VerDec$ scales only with the size of the plaintext, and \emph{not} with the size of the circuit performed via $\Eval$. 
We do require that $\VerDec$ performs a classical computation which can scale with the size of the circuit; this is reasonable since $\VerDec$ must receive the circuit as input. 
Like the other currently-known schemes for QFHE, our scheme is leveled, in the sense that pre-generated auxiliary magic states are needed to perform the evaluation procedure.
\begin{theorem}[Main result, informal]
Let $\TTP$ be the scheme outlined above, and let $\VerDec^{\equiv}$ be $\VerDec$ for the case of verifying the empty circuit.
\begin{enumerate}
\item The vQFHE scheme $\TTP$ satisfies IND-VER security.
\item The scheme $(\KeyGen, \Enc, \VerDec^{\equiv})$ is authenticating~\cite{DNS12} and IND-CPA~\cite{BJ15}.
\end{enumerate}
\end{theorem}

\paragraph{Application: quantum one-time programs.}

A one-time program (or OTP) is a device which implements a circuit, but self-destructs after the first use. OTPs are impossible without hardware assumptions, even with quantum states; OTPs that implement quantum circuits (qOTP) can be built from classical OTPs (cOTP)~\cite{BGS13}. As a first application of vQFHE, we give another simple construction of qOTPs. Our construction is weaker, since it requires a computational assumption. On the other hand, it is conceptually very simple and serves to demonstrates the power of verification. In our construction, the qOTP for a quantum circuit $C$ is simply a (vQFHE) encryption of $C$ together with a cOTP for verifying the universal circuit. To use the resulting qOTP, the user attaches their desired input, homomorphically evaluates the universal circuit, and then plugs their computation log into the cOTP to retrieve the final decryption keys.

\subsubsection{Preliminaries.}

Our exposition assumes a working knowledge of basic quantum information and the associated notation. As for the particular notation of quantum gates, the gates $(\H, \P, \CNOT)$ generate the so-called Clifford group (which can also be defined as the normalizer of the Pauli group); it includes the Pauli gates $\X$ and $\Z$. In order to implement arbitrary unitary operators, it is sufficient to add the $\T$ gate (also known as the $\pi/8$ gate). Finally, we can reach universal quantum computation by adding single-qubit measurements in the computational basis.

We will frequently make use of several standard cryptographic ingredients, as follows. The quantum one-time pad (QOTP) will be used for information-theoretically secret one-time encryption. In its encryption phase, two bits $a, b \in \{0,1\}$ are selected at random, and the map $\X^a \Z^b$ is applied to the input, projecting it to the maximally-mixed state. We will also need the computational security notions for quantum secrecy, including indistinguishability (IND, IND-CPA)~\cite{BJ15} and semantic security (SEM)~\cite{ABFGSS16}. For quantum authentication, we will refer to the security definition of Dupuis, Nielsen and Salvail~\cite{DNS12}. We will also make frequent use of the trap code for quantum authentication, described below in Section~\ref{sec:trap-code-encrypt}. For a security proof and methods for interactive computation on this code, see~\cite{BGS13}. Finally, we will also use classical fully-homomorphic encryption (FHE). In brief, an FHE scheme consists of classical algorithms $(\KeyGen, \Enc, \Eval, \Dec)$ for (respectively) generating keys, encrypting plaintexts, homomorphically evaluating circuits on ciphertexts, and decrypting ciphertexts. We will use FHE schemes which are quantum-secure and whose $\Dec$ circuits are in $\NC^1$ (see, e.g.,~\cite{BV11}).

\section{A new primitive: verifiable QFHE}

We now define verified quantum fully-homomorphic encryption (or vQFHE), in the symmetric-key setting. The public-key case is a straightforward modification. 

\subsubsection{Basic definition.}

The definition has two parameters: the class $\mathcal C$ of circuits which the user can verify, and the class $\mathcal V$ of circuits which the user needs to perform in order to verify. We are interested in cases where $\mathcal C$ is stronger than $\mathcal V$.

\begin{definition}[vQFHE]\label{def:vQFHE} Let $\mathcal C$ and $\mathcal V$ be (possibly infinite) collections of quantum circuits. A $(\mathcal C, \mathcal V)$-vQFHE scheme is a set of four QPT algorithms:
\begin{itemize}
\item $\KeyGen : \{1\}^{\kappa} \to \mathcal K \times \states(\mathcal H_E)$ (security parameter $\to$ private key, eval key);
\item $\Enc : \mathcal K \times \states(\hi_X) \to \states(\hi_C)$ (key, ptext $\to$ ctext);
\item $\Eval : \mathcal C \times \states(\hi_{CE}) \to \mathcal L \times \states(\hi_C)$ (circuit, eval key, ctext $\to$ log,  ctext);
\item $\VerDec: \mathcal K \times \mathcal C \times \mathcal L \times \states(\hi_C) \to \states(\hi_X)\times \{\cacc, \crej\}$  
\end{itemize}
such that (i.) the circuits of $\VerDec$ belong to the class $\mathcal V$, and (ii.) for all $(sk, \rho_\evk) \from \KeyGen$, all circuits $c \in \mathcal C$, and all $\rho \in \states(\hi_{XR})$,
$$
\bigl\|\VerDec_{sk} (c, \Eval(c, \Enc_k(\rho), \rho_\evk)) - \Phi_c(\rho) \otimes \egoketbra{\cacc}) \bigr\|_1 \leq \negl(\kappa)\,,
$$
where $R$ is a reference and the maps implicitly act on appropriate spaces.
\end{definition}

We will refer to condition (ii.) as \emph{correctness}. It is implicit in the definition that the classical registers $\mathcal K, \mathcal L$ and the quantum registers $E, X, C$ are really infinite families of registers, each consisting of $\poly(\kappa)$-many (qu)bits. In some later definitions, it will be convenient to use a version of $\VerDec$ which also outputs a copy of the (classical) description of the circuit $c$.

\paragraph{Compactness.}

We note that there are trivial vQFHE schemes for some choices of $(\mathcal C, \mathcal V)$ (e.g., if $\mathcal C \subset \mathcal V$, then the user can simply authenticate the ciphertext and then perform the computation during decryption). Earlier work on quantum and classical homomorphic encryption required compactness, meaning that the size of the decrypt circuit should not scale with the size of the homomorphic circuit.

\begin{definition}[Compactness of QFHE]\label{def:compactness-qfhe}
A QFHE scheme $S$ is compact if there exists a polynomial $p(\kappa)$ such that for any circuit $C$ with $n_{\mathsf{out}}$ output qubits, and for any input $\rho_X$, the complexity of applying $S.\Dec$ to $S.\Eval^C(S.\Enc_{sk}(\rho_X),\rho_{evk})$ is at most $p(n_{\mathsf{out}},\kappa)$.
\end{definition}

When considering QFHE \emph{with} verification, however, some tension arises. On one hand, trivial schemes like the above still need to be excluded. On the other hand, verifying that a circuit $C$ has been applied requires reading a description of $C$, which violates Definition~\ref{def:compactness-qfhe}. We thus require a more careful consideration of the relationship between the desired circuit $C \in \mathcal C$ and the verification circuit $V \in \mathcal V$. In our work, we will allow the number of classical gates in $V$ to scale with the size of $C$. We propose a new definition of compactness in this context.

\begin{definition}[Compactness of vQFHE (informal)]\label{def:compactness-vqfhe-informal}
A vQFHE scheme $S$ is compact if $S.\VerDec$ is divisible into a classical verification procedure $S.\Ver$ (outputting only an accept/reject flag), followed by a quantum decryption procedure $S.\Dec$. The running time of $S.\Ver$ is allowed to depend on the circuit size, but the running time of $S.\Dec$ is not.
\end{definition}

The procedure $S.\Dec$ is not allowed to receive and use any other information from $S.\Ver$ than whether or not it accepts or rejects. This prevents the classical procedure $S.\Ver$ from de facto performing part of the decryption work (e.g., by computing classical decryption keys). In Section~\ref{sec:trap-code-encrypt}, we will see a scheme that does not fulfill compactness for this reason.

\begin{definition}[Compactness of vQFHE (formal)]\label{def:compactness-vqfhe}
A vQFHE scheme $S$ is compact if there exists a polynomial $p$ such that $S.\VerDec$ can be written as $S.\Dec \circ\, S.\Ver$, and the output ciphertext space $\states(\hi_C)$ can be written as a classical-quantum state space $\mathcal{A} \times \states(\hi_B)$, where (i.) $S.\Ver : \mathcal{K} \times \mathcal{C} \times \mathcal{L} \times \mathcal{A} \to \{\cacc, \crej\}$ is a classical polynomial-time algorithm, and (ii.) $S.\Dec : \{\cacc, \crej\} \times \mathcal{K} \times \states(\hi_C) \to \states(\hi_X) \times \{\cacc,\crej\}$ is a quantum algorithm such that for any circuit $C$ with $n_{\mathsf{out}}$ output qubits and for any input $\rho_X$, $\S.\Dec$ runs in time $p(n_{\mathsf{out}},\kappa)$ on the output of $S.\Eval^C(S.\Enc(\rho_X),\rho_{evk})$.
\end{definition}
Note that in the above definition, the classical registers $\mathcal{K}$ and $\mathcal{A}$ are copied and fed to both $S.\Dec$ and $S.\Ver$.

\label{sec:def-privacy}
For privacy, we say that a vQFHE scheme is private if its ciphertexts are indistinguishable under chosen plaintext attack (IND-CPA)~\cite{BJ15, DSS16}.

\subsubsection{Secure verifiability.}

In this section, we formalize the concept of verifiability. Informally, one would like the scheme to be such that whenever $\VerDec$ accepts, the output can be trusted to be close to the desired output. We will consider two formalizations of this idea: a semantic one, and an indistinguishability-based one.

Our semantic definition will state that every adversary with access to the ciphertext can be simulated by a simulator that only has access to an ideal functionality that simply applies the claimed circuit. It is inspired by quantum authentication~\cite{DNS12, BW16} and semantic secrecy~\cite{ABFGSS16}.

The real-world scenario (Figure~\ref{fig:SEM-VER}, top) begins with a state $\rho_{XR_1R_2}$ prepared by a QPT (``message generator'') $\mathcal{M}$. The register $X$ (plaintext) is subsequently encrypted and sent to the adversary $\advA$. The registers $R_1$ and $R_2$ contain side information. The adversary acts on the ciphertext and $R_1$, producing some output ciphertext $C_{X'}$, a circuit description $c$, and a computation log $log$. These outputs are then sent to the verified decryption function. The output, along with $R_2$, is sent to a distinguisher $\disD$, who produces a bit 0 or 1.

\begin{figure}
	\begin{center}
	\scalebox{0.8}{
	\makebox[\textwidth][c]{
		\begin{tikzpicture}
		\draw (-0.5,1.5) -- (0,1.5);
		\node[anchor=east] at (-0.5,1.5) {$\rho_{evk}$};
		\draw (0,-1) rectangle (1,4);
		\node at (0.5,1.5) {$\inM$};
		
		\draw (1,0.5) -- (4,0.5);
		\node[anchor=south west] at (,0.5) {$R_1$};
		\draw (1,-0.5) -- (10,-0.5);
		\node[anchor=south west] at (1,-0.5) {$R_2$};
		\draw (1,3.5) -- (2,3.5);
		\node[anchor=south west] at (1,3.5) {$X$};
		
		\draw (2,3) rectangle (3,4);
		\node at (2.5,3.5) {$\Enc_{sk}$};
		\draw (3,3.5) -- (4,3.5);
		\node[anchor=south] at (3.4,3.5) {$C_X$};
		
		\draw (4,0) rectangle (5,4);
		\node at (4.5,2) {$\advA$};
		
		\draw (5,3.5) -- (6.5,3.5);
		\node[anchor=south west] at (5,3.5) {$C_{X'}$};
		\draw (5,2.5) -- (6.5,2.5);
		\draw (5,2.6) -- (6.5,2.6);
		\node[anchor=south west] at (5,2.55) {$c$};
		\draw (5,1.6) -- (6.5,1.6);
		\draw (5,1.5) -- (6.5,1.5);
		\node[anchor=south west] at (5,1.55) {$log$};
		\draw (5,0.5) -- (10,0.5);
		\node[anchor=south west] at (5,0.5) {$R_1'$};
		
		\draw (6.5,1) rectangle (7.5,4);
		\node[rotate=90] at (7,2.5) {$\VerDec_{sk}$};
		
		\draw (7.5,3.5) -- (10,3.5);
		\node[anchor=south west] at (7.5,3.5) {$X'$};
		\draw (7.5,2.5) -- (10,2.5);
		\draw (7.5,2.6) -- (10,2.6);
		\node[anchor=south west] at (7.5,2.55) {$c$};
		\draw (7.5,1.5) -- (10,1.5);
		\draw (7.5,1.6) -- (10,1.6);
		\node[anchor=south west] at (7.5,1.55) {$acc/rej$};
		
		\draw (10,-1) rectangle (11,4);
		\node at (10.5,1.5) {$\disD$};
		\draw (11,1.5) -- (11.5,1.5);
		\draw (11,1.6) -- (11.5,1.6);
		\node[anchor=west] at (11.5,1.5) {$0/1$};
		\end{tikzpicture}
	}}
	\end{center}

	\begin{center}
	\scalebox{0.8}{
	\makebox[\textwidth][c]{
		\begin{tikzpicture}
	\draw (-0.5,1.5) -- (0,1.5);
	\node[anchor=east] at (-0.5,1.5) {$\rho_{evk}$};
	\draw (0,-1) rectangle (1,4);
	\node at (0.5,1.5) {$\inM$};
	
	\draw (1,0.5) -- (4,0.5);
	\node[anchor=south west] at (,0.5) {$R_1$};
	\draw (1,-0.5) -- (10,-0.5);
	\node[anchor=south west] at (1,-0.5) {$R_2$};
	\draw (1,3.5) -- (6.5,3.5);
	\node[anchor=south west] at (1,3.5) {$X$};
	
	\draw (4,0) rectangle (5,3);
	\node at (4.5,1.5) {$\simS_{sk}$};
	
	\draw (5,2.5) -- (10,2.5);
	\draw (5,2.6) -- (6.7,2.6) -- (6.7,3);
	\draw (6.8,3) -- (6.8,2.6) -- (10,2.6);
	\node[anchor=south] at (5.4,2.55) {$c$};
	\draw (5,1.6) -- (10,1.6);
	\draw (5,1.5) -- (10,1.5);
	\node[anchor=south west] at (5,1.55) {$acc(0)/rej(1)$};
	\draw (5,0.5) -- (10,0.5);
	\node[anchor=south] at (5.25,0.5) {$R_1'$};
	
	\draw (6.5,3) rectangle (7.5,4);
	\node at (7,3.5) {$\Phi_c$};
	\draw (7.5,3.5) -- (8.5,3.5);
	\node[anchor=south west] at (7.5,3.5) {$X'$};
	
	\node at (9,1.55) {$\bullet$};
	\draw (9,1.6) -- (9,3);
	\draw (8.5,3) rectangle (9.5,4);
	\node at (9,3.5) {\faBan};
	\draw (9.5,3.5) -- (10,3.5);
	
	\draw (10,-1) rectangle (11,4);
	\node at (10.5,1.5) {$\disD$};
	\draw (11,1.5) -- (11.5,1.5);
	\draw (11,1.6) -- (11.5,1.6);
	\node[anchor=west] at (11.5,1.5) {$0/1$};
	\end{tikzpicture}
	}}
	\end{center}
		\caption{The real-world (top) and ideal-world (bottom) for SEM-VER.}
		\label{fig:SEM-VER}
\end{figure}
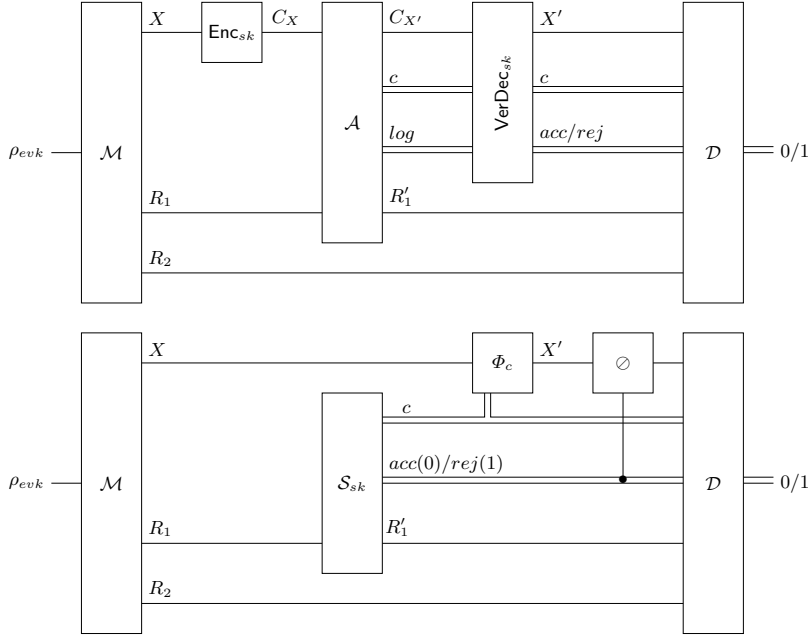

In the ideal-world scenario (Figure~\ref{fig:SEM-VER}, bottom), the plaintext $X$ is not encrypted or sent to the simulator $\simS$. The simulator outputs a circuit $c$ and chooses whether to accept or reject. The channel $\Phi_c$ implemented by $c$ is applied to the input register $X$ directly. If reject is chosen, the output register $X'$ is traced out and replaced by the fixed state $\Omega$; this controlled-channel is denoted $\barctrl \faBan$.

\begin{definition}[$\kappa$-SEM-VER]\label{def:SEM-VER}
	A vQFHE scheme $(\KeyGen, \Enc, \Eval, \VerDec)$ is \emph{semantically $\kappa$-verifiable} if for any QPT adversary $\advA$, there exists a QPT $\simS$ such that for all QPTs $\inM$ and $\disD$,
$$
\bigg|\Pr\Big[\disD\Big(\Real^\advA_{sk} (\inM (\rho_{evk}))\Big) = 1\Big]
- \Pr\Big[\disD\Big(\Ideal^\simS_{sk}(\inM(\rho_{evk})) \Big) = 1\Big]\bigg|
\leq \negl(\kappa),
$$
where
$\Real^\advA_{sk} = \VerDec_{sk} \circ \advA \circ \Enc_{sk}$
and
$\Ideal^\simS_{sk} = \barctrl \faBan \circ\, \Phi_c \circ \simS_{sk}$,
and the probability is taken over $(\rho_\evk,sk) \leftarrow \KeyGen(1^{\kappa})$ and all QPTs above.
\end{definition}

Note that the simulator (in the ideal world) gets the secret key $sk$. We believe that this is necessary, because the actions of an adversary may depend on superficial properties of the ciphertext. In order to successfully simulate this, the simulator needs to be able to generate (authenticated) ciphertexts. He cannot do so with a fresh secret key, because the input plaintext may depend on the correlated evaluation key $\rho_{evk}$. Fortunately, the simulator does not become too powerful when in possession of the secret key, because he does not receive any relevant plaintexts or ciphertexts to encrypt or decrypt: the input register $X$ is untouchable for the simulator.

Next, we present an alternative definition of verifiability, based on a security game motivated by indistinguishability.

\begin{game}\label{def:VER-game}
	For an adversary $\advA = (\advA_1, \advA_2, \advA_3)$, a scheme $S$, and a security parameter $\kappa$, the $\VerGame{S}$ game proceeds as depicted in Figure~\ref{fig:vergame}.
	
	\begin{figure}
		\centering
	\scalebox{0.8}{	\makebox[\textwidth][c]{
		\begin{tikzpicture}
		\draw (0,0) rectangle (1,3);
		\node[rotate=90] at (0.5,1.5) {$S.\KeyGen(1^{\kappa})$};
		\draw (1,2) -- (2,2);
		\draw (1,0.1) -- (1.2,0.1);
		\draw (1,0.2) -- (1.2,0.2);
		\node[anchor=south] at (1.5,2) {$\rho_{evk}$};
		\node[anchor=west] at (1.2,0.2) {$sk$};
		
		\node at (2.5,1.875) {$\advA_1$};
		\draw (2,0) rectangle (3,3.75);
		\draw (3,0.5) -- (6,0.5);
		\draw (3,3.5) -- (6,3.5);
		\node[anchor=south west] at (3,0.5) {$R$};
		\node[anchor=south west] at (3,3.5) {$X$};
		\node[anchor=east] at (3,4.5) {$|0^n\rangle\langle0^n|$};
		\node[anchor=east] at (3,5.5) {$r \in_R \{0,1\}$};
		\draw (3,5.5) -- (11.5,5.5);
		\draw (3,5.6) -- (11.5,5.6);
		\draw (3,4.5) -- (11.5,4.5);
		
		\filldraw[fill=white] (4.3,3.25) rectangle (5.7,3.75);
		\node at (5,3.5) {$S.\Enc_{sk}$};
		\node at (4,5.53) {$\bullet$};
		\draw (4,5.5) -- (4,3.5);
		\node at (4,4.5) {$\times$};
		\node at (4,3.5) {$\times$};
		
		\filldraw[fill=white] (6,0) rectangle (7,3.75);
		\node at (6.5,1.875) {$\advA_2$};
		\draw (7,0.5) -- (11,0.5);
		\draw (7,1.5) -- (8,1.5);
		\draw (7,1.6) -- (8,1.6);
		\draw (7,2.5) -- (8,2.5);
		\draw (7,2.6) -- (7.7,2.6) -- (7.7,4);
		\draw (7.8,4) -- (7.8,2.6) -- (8,2.6);
		\draw (7,3.5) -- (7.7,3.5);
		\draw (7.8,3.5) -- (8,3.5);
		\node[anchor=south west] at (6.9,3.5) {$C_{X'}$};
		\node[anchor=south west] at (7,2.55) {$c$};
		\node[anchor=south west] at (7,1.55) {$log$};
		\node[anchor=south west] at (7,0.5) {$R'$};
		
		\filldraw[fill=white] (7.5,4) rectangle (8.5,5);
		\node at (8,4.5) {$\Phi_c$};

		\draw (8,1) rectangle (9,3.75);
		\node[rotate=90] at (8.5,2.375) {$S.\VerDec_{sk}$};
		\draw (9,1.5) -- (11,1.5);
		\draw (9,1.6) -- (11,1.6);
		\draw (9,2.5) -- (11,2.5);
		\draw (9,2.6) -- (11,2.6);
		\draw (9,3.5) -- (11,3.5);
		\node[anchor=north] at (10,1.6) {$acc/rej$};
		\node[anchor=south west] at (9,2.55) {$c$};
		\node[anchor=south west] at (9,3.5) {$X'$};
		
		\node at (9.75,1.55) {$\bullet$};
		\draw (9.75,1.6) -- (9.75,4.5);
		\filldraw[fill=white] (9.25,4) rectangle (10.25,5);
		\node at (9.75,4.5) {\faBan};
		\node at (10.5,5.53) {$\bullet$};
		\draw (10.5,5.5) -- (10.5,3.5);
		\node at (10.5,4.5) {$\times$};
		\node at (10.5,3.5) {$\times$};
		
		\draw (11,0) rectangle (12,3.75);
		\node at (11.5,1.875) {$\advA_3$};
		\draw (12,1.875) -- (12.5,1.875);
		\draw (12,1.975) -- (12.5,1.975);
		\node[anchor=west] at (12.5,1.975) {$r'$};
		\end{tikzpicture}
	}}
	\caption{The indistinguishability game $\VerGame{S}$, as used in the definition of $\kappa$-IND-VER.}
	\label{fig:vergame}
	\end{figure}
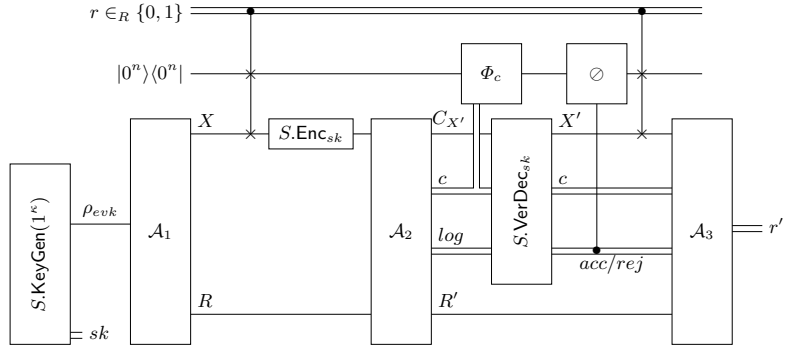
\end{game}
The game is played in several rounds. Based on the evaluation key, the adversary first chooses an input (and some side information in $R$). Based on a random bit $b$ this input is either encrypted and sent to $\advA_2$ (if $b=0$), or swapped out and replaced by a dummy input $\ketbra{0^n}{0^n}$ (if $b=1$). If $b=1$, the ideal channel $\Phi_c$ is applied by the challenger, and the result is swapped back in right before the adversary (in the form of $\advA_3$) has to decide on its output bit $b'$. If $\advA_2$ causes a reject, the real result is also erased by the channel \faBan. We say that the adversary \emph{wins} (expressed as $\VerGame{S} = 1$) whenever $b' = b$. 

\begin{definition}[$\kappa$-IND-VER]
	A vQFHE scheme $S$ has \emph{$\kappa$-indistinguishable verification} if for any QPT adversary $\advA$,
$
	\Pr[\VerGame{S} = 1] \leq \frac{1}{2} + \negl(\kappa).
$
\end{definition}

\begin{theorem}\label{thm:IND-SEM}
	A vQFHE scheme is $\kappa$-IND-VER iff it is $\kappa$-SEM-VER.
\end{theorem}
\begin{proof}
We first show the forward direction. Suppose a scheme $S$ is \emph{not} $\kappa$-SEM-VER. Then there exists a QPT $\advA$ such that for all simulators $\simS$, there exist QPTs $\inM$ and $\disD$ and a polynomial $p$ such that the difference in acceptance probability is at least $1/p(\kappa)$. Choose $\simS$ to be
	\[ 
	\simS : (sk, \rho_{R_1})
	\mapsto
	\tr_{X'}\Big(
		(\VerDec_{sk} \otimes \id_{R_1})
		(\advA(\Enc_{sk}(\ketbra{0^n}{0^n}) \otimes \rho_{R_1}))
	\Big),
	\]
This simulator encrypts a dummy state and feeds it to the adversary; whatever comes out is then checked. Note that in the accept case, the output is wrong, since the claimed circuit is applied to the dummy state instead of the real input. This does not matter, however, because the simulator throws out the result immediately. Since $\algo S$ is a possible simulator, we can let $\inM$ and $\disD$ be as given by the assumption that $\kappa$-SEM-VER is false. 
	
	This allows us to construct a QPT adversary $\advA' = (\advA'_1,\advA'_2,\advA'_3)$ for the VER indistinguishability game $\VerGame{S}$ simply by setting $\advA'_1 = \inM$, $\advA'_2 = (\advA \otimes \id_{R_2})$, and $\advA'_3 = \disD$. Informally, the probability that this adversary wins is
	\[
	\Pr[r = 0]\Pr[\advA'_3 \text{ guesses }0 \mid r = 0] \ \ + \ \  \Pr[r=1]\Pr[\advA'_3 \text{ guesses }1 \mid r = 1]\,.
	\]
More precisely, it is 
	\begin{align*}
	&\frac{1}{2} \Pr\Big[
	\advA'_3\Big(
		(S.\VerDec_{sk} \otimes \id_R)
		(\advA'_2((S.\Enc_{sk} \otimes \id_R)(
			\advA'_1(\rho_{evk})
		)))
	\Big) = 0
	\Big]\\
	+ &\frac{1}{2} \Pr\Big[
		\advA'_3 \Big(
			\big(\id_{X'R'} \otimes \Pi_\cacc \big)
			\big((\Phi_c \otimes \id_{FR'})(\sigma_{XFR'})\big) + \big(\id_{X'R'} \otimes \Pi_\crej \big)
	\big(\Omega \otimes \sigma_{FR'} \big)\Big) = 1
	\Big]
	\end{align*}
	where $F$ is the flag register (accept/reject), and we set $(c, \sigma_{XFR'}) = (\tr_{X'} \circ (\id_X \otimes S.\VerDec_{sk} \otimes \id_{R'}) \circ (\id_X \otimes \advA'_2))(S.\Enc_{sk}(\ketbra{0^n}{0^n}) \otimes \advA'_1(\rho_{evk}))$, and $\Pi_\cacc = \accstate$ and $\Pi_\crej = \rejstate$. This can be seen by following the wires in the indistinguishability game. By our definition of $\advA'$ and $\simS$, this is equal to
	\begin{align*}
	&\frac{1}{2} \Big(1 -
	\Pr\Big[
	\disD\Big(
	\big(S.\VerDec_{sk} \otimes \id_{R'_1R_2}\big)
	\big((\advA \otimes \id_{R_2}) \tau_{C_XR}
	\big)\Big) = 1
	\Big]
	\Big)\\
	+ &\frac{1}{2}\Pr\Big[\disD\Big((\barctrl \faBan \circ\, \Phi_c \circ\, \simS_{sk} ) (\inM(\rho_{evk}))\Big) = 1\Big]
	\end{align*}
where $\tau_{C_XR} = (S.\Enc_{sk} \otimes \id_{R_1R_2})(\inM(\rho_{evk}))$.
	By the assumption that $S$ is not $\kappa$-SEM-VER, this is at least $\frac{1}{2} + 1/p(\kappa)$. Hence, this adversary wins the IND-VER indistinguishability game with nonnegligible probability. 

The reverse direction of the main claim is relatively straightforward. From an arbitrary adversary $\algo A$ for the IND-VER indistinguishability game, we define a semantic adversary, message generator, and distinguisher, that together simulate the game for $\algo A$. The fact that $S$ is $\kappa$-SEM-VER allows us to limit the advantage of the semantic adversary over any simulator, and thereby the winning probability of $\algo A$. For a detailed proof, see Appendix \ref{app:IND-SEM}.
\qed
\end{proof}

\section{\texorpdfstring{$\TC$}{TC}: A partially-homomorphic scheme with verification}
We now present a partially-homomorphic scheme with verification, which will serve as a building block for the fully-homomorphic scheme in Section~\ref{sec:ttp-definition}. It is called $\TC$ (for ``trap code''), and is homomorphic only for $\CNOT$, (classically controlled) Paulis, and measurement in the computational and Hadamard basis. It does not satisfy compactness: as such, it performs worse than the trivial scheme where the client performs the circuit at decryption time. However, $\TC$ lays the groundwork for the vQFHE scheme we present in Section~\ref{sec:ttp-definition}, and as such is important to understand in detail. It is a variant of the trap-code scheme presented in~\cite{BGS13} (which requires classical interaction for $\T$ gates), adapted to our vQFHE framework. A variation also appears in~\cite{BJSW16}, and implicitly in~\cite{SP00}. 

\subsubsection{Setup and encryption.}\label{sec:trap-code-encrypt}

Let $\CSS$ be a (public) self-dual $[[m,1,d]]$ CSS code, so that $\H$ and $\CNOT$ are transversal. $\CSS$ can correct $d_c$ errors, where $d=2d_c + 1$. We choose $m=\poly(d)$ and large enough that $d_c=\kappa$ where $\kappa$ is the security parameter. The concatenated Steane code satisfies all these requirements. 

We generate the keys as follows. Choose a random permutation $\pi \in_{R} S_{3m}$ of $3m$ letters. Let $n$ be the number of qubits that will be encrypted. For each $i \in \{1, \dots, n\}$, pick bit strings $x[i] \in_{R} \{0,1\}^{3m}$ and $z[i] \in_{R} \{0,1\}^{3m}$. The secret key $sk$ is the tuple $(\pi, x[1], z[1], \dots, x[n], z[n])$, and $\rho_{\evk}$ is left empty.
 
Encryption is per qubit: (i.) the state $\sigma$ is encoded using $\CSS$, (ii.) $m$ computational and $m$ Hadamard `traps' ($\ket0$ and $\ket{+}$ states, see~\cite{BGS13}) are added, (iii.) the resulting $3m$ qubits are permuted by $\pi$, and (iv.) the overall state is encrypted with a quantum one-time pad (QOTP) as dictated by $x = x[i]$ and $z = z[i]$ for the $i$th qubit. We denote the ciphertext by $\enc{\sigma}$. See Algorithm \ref{alg:tc-enc} for details.

\subsubsection{Evaluation.}\label{sec:tc-eval}

First, consider Pauli gates. By the properties of $\CSS$, applying a logical Pauli is done by applying the same Pauli to all physical qubits. The application of Pauli gates ($\X$ and/or $\Z$) to a state encrypted with a quantum one-time pad can be achieved without touching the actual state, by updating the keys to QOTP in the appropriate way. This is a classical task, so we can postpone the application of the Pauli to $\VerDec$ (recall it gets the circuit description) without giving up compactness for $\TC$. So, formally, the evaluation procedure for Pauli gates is the identity map. Paulis conditioned on a classical bit $b$ which will be known to $\VerDec$ at execution time (e.g., a measurement outcome) can be applied in the same manner. 

Next, we consider $\CNOT$.\label{sec:tc-eval-cnot} To apply a $\CNOT$ to encrypted qubits $\sigma_i$ and $\sigma_j$, we apply $\CNOT$ transversally between the $3m$ qubits of $\enc{\sigma_i}$ and the $3m$ qubits of $\enc{\sigma_j}$. Ignoring the QOTP for the moment, the effect is a transversal application of $\CNOT$ on the pysical data qubits (which, by $\CSS$ properties, amounts to logical $\CNOT$ on $\sigma_i \otimes \sigma_j$), and an application of $\CNOT$ between the $2m$ pairs of trap qubits. Since $\CNOT\ket{00} = \ket{00}$ and $\CNOT\ket{++} = \ket{++}$, the traps are unchanged. Note that $\CNOT$ commutes with the Paulis that form the QOTP. In particular, for all $a,b,c,d \in \{0,1\}$,
$
\CNOT (\X^a_1\Z^b_1 \otimes \X^c_2\Z^d_2) = (\X^a_1\Z^{b \oplus d}_1 \otimes \X^{a \oplus c}_2\Z^d_2)\CNOT.
$
Thus, updating the secret-key bits $(a,b,c,d)$ to $(a,b\oplus d,a\oplus c, d)$ finishes the job. The required key update happens in $\TC.\VerDec$ (see Algorithm~\ref{alg:tc-verdec}). 

Next, consider computational-basis measurements. For $\CSS$, logical measurement is performed by measurement of all physical qubits, followed by a classical decoding procedure~\cite{BGS13}. In $\TC.\Eval$, we measure all $3m$ ciphertext qubits. During $\TC.\VerDec$, the contents of the measured qubits (now a classical string $a \in \{0,1\}^{3m}$) will be interpreted into a logical measurement outcome.

Finally, we handle Hadamard-basis measurements. A transversal application of $\H$ to all $3m$ relevant physical qubits precedes the evaluation procedure for the computational basis measurement. Since $\CSS$ is self-dual, this applies a logical $\H$. Since $\H\ket0 = \ket+$ and $\H\ket+ = \ket0$, all computational traps are swapped with the Hadamard traps. This is reflected in the way $\TC.\VerDec$ checks the traps (see Algorithm~\ref{alg:tc-dec-measurement}). Note that Algorithm~\ref{alg:tc-dec-measurement} is a classical procedure (and thus its accept/reject output flags is classical). $\CSS.\Classical\Decode$ is the classical decoding procedure that calculates the logical measurement outcome~\cite{BGS13}.

\subsubsection{Verification and decryption.}

If a qubit is unmeasured after evaluation (as stated in the circuit), $\TC.\VerDecQubit$ is applied (Algorithm~\ref{alg:tc-dec-qubit}). This removes the QOTP, undoes the permutation, checks all traps, and decodes the qubit.

If a qubit is measured during evaluation, $\TC.\VerDec$ receives a list $\enc{w}$ of $3m$ physical measurement outcomes for that qubit. These outcomes are classically processed (removing the QOTP by flipping bits, undoing $\pi$, and decoding $\CSS$) to produce the plaintext measurement outcome (see Algorithm~\ref{alg:tc-dec-measurement}). Note that we only checks the $\ket0$ traps in this case. Intuitively, this should not affect security, since any attack that affects only $\ket{+}$ but not $\ket0$ will be canceled by computational basis measurement.Note that Algorithm~\ref{alg:tc-dec-measurement} is a completely classical procedure. $\CSS.\Classical\Decode$ is the classical decoding procedure that recovers the logical measurement outcome from the physical ones~\cite{BGS13}.

The complete procedure $\TC.\VerDec$ updates the QOTP keys according to the gates in the circuit description, and then decrypts all qubits and measurement results as described above (see Algorithm~\ref{alg:tc-verdec}).

\begin{algorithm}
	\caption{$\TC.\Enc((\pi, x, z), \sigma)$}\label{alg:tc-enc}
	\begin{algorithmic}[1]
		\State $\sigma' \gets \CSS.\Encode(\sigma)$
		\State $\sigma'' \gets \textbf{permute}_{\pi}(\sigma' \otimes \ket{0^m} \otimes \ket{+^m})$
		\State $\enc{\sigma} \gets \displaystyle \X^{x} \Z^{z} \sigma'' \X^{x}\Z^{z} $
		\State \textbf{return} $\enc{\sigma}$
	\end{algorithmic}
\end{algorithm}

\begin{algorithm}
	\caption{$\TC.\VerDecQubit((\pi, x,z), \enc{\sigma})$}\label{alg:tc-dec-qubit}
	\begin{algorithmic}[1]
		\State $\sigma \gets \X^{x} \Z^{z} \enc{\sigma} \X^{x}\Z^{z}$
		\State $(\sigma', \text{traps}_X, \text{traps}_Z) \gets \textbf{permute}_{\pi^{-1}}(\sigma)$\Comment{traps$_X$ and traps$_Z$ are quantum states}
		\State Measure traps$_X$ in the computational basis. If not all-zero, \textbf{return} $(\Omega, \rejstatepure)$.
		\State Measure traps$_Z$ in the Hadamard basis. If not all-+, \textbf{return} $(\Omega, \rejstatepure)$.
		\State $\sigma'' \gets \CSS.\Decode($first $m$ qubits of $\sigma')$
		\State \textbf{return} $(\sigma'', \accstatepure)$
	\end{algorithmic}
\end{algorithm}

\begin{algorithm}
	\caption{$\TC.\VerDecMeasurement((\pi, x,z), \enc{w}, basis)$}\label{alg:tc-dec-measurement}
	\begin{algorithmic}[1]
		\If{$basis=+$}
		\State $w' \gets \textbf{permute}_{\pi^{-1}}(x \oplus \enc{w})$
		\State Check the second $m$ bits of $w'$. If not all-zero, \textbf{return} $(0, \texttt{rej})$.
		\EndIf
		\If{$basis=\times$}
		\State $w' \gets \textbf{permute}_{\pi^{-1}}(z \oplus \enc{w})$
		\State Check the third $m$ bits of $w'$. If not all-zero, \textbf{return} $(0, \texttt{rej})$.
		\EndIf
		\State $w'' \gets \CSS.\Classical\Decode($first $m$ bits of $w')$
		\State \textbf{return} $(w'', \texttt{acc})$
	\end{algorithmic}
\end{algorithm}

\begin{algorithm}
\caption{$\TC.\VerDec((\pi,x,z),\enc{\sigma},c)$}\label{alg:tc-verdec}
\begin{algorithmic}[1]
\ForAll{gates $G$ in $c$}
	\If{$G = X_i$}
		\State $x[i] \gets x[i] \oplus \textbf{permute}_{\pi}(1^m0^{2m})$\Comment{update keys (see Section~\ref{sec:tc-eval})}
	\ElsIf{$G = Z_i$}
		 \State $z[i] \gets z[i] \oplus \textbf{permute}_{\pi}(1^m0^{2m})$\Comment{update keys (see Section~\ref{sec:tc-eval})}
	\ElsIf{$G = \CNOT$}
		\State $(x[i], z[i]) (x[j], z[j]) \gets (x[i], z[i]\oplus z[j]) (x[i] \oplus x[j], z[j])$\Comment{update keys}
	\ElsIf{$G$ is a measurement in basis $b$ on qubit $i$}
		\State $(a_i,\mathit{flag}) \gets \TC.\VerDecMeasurement((\pi,x[i],z[i]),\enc{\sigma_i},b)$
		\If{$\mathit{flag} = \texttt{rej}$}
			\State  \textbf{return} $(\Omega, \rejstatepure)$.
		\EndIf
	\EndIf
\EndFor
\State Execute $\TC.\VerDecQubit$ on all unmeasured qubits. If it rejects, \textbf{return} $(\Omega, \rejstatepure)$.
\State $\sigma' \gets$ the list of decrypted qubits (and measurement outcomes $a_i$).
\State $\sigma'' \gets \sigma'$ with all wires that are not part of the output
		of $c$ traced out.
\State \textbf{return} $(\sigma'', \accstatepure)$
		
\end{algorithmic}
\end{algorithm}

\subsubsection{Correctness, compactness, and privacy.}

For honest evaluation, $\TC.\VerDec$ accepts with probability 1. Correctness is straightforward to check by following the description in Section~\ref{sec:tc-eval}. For privacy, note that the final step in the encryption procedure is the application of a (information-theoretically secure) QOTP with fresh, independent keys. If IND-CPA security is desired, one could easily extend $\TC$ by using a pseudorandom function for the QOTP, as in~\cite{ABFGSS16}.

$\TC$ is not compact in the sense of Definition~\ref{def:compactness-vqfhe}, however. In order to compute the final decryption keys, the whole gate-by-gate key update procedure needs to be executed, aided by the computation log and information about the circuit. Thus, we cannot break $\TC.\VerDec$ up into two separate functionalities, $\Ver$ and $\Dec$, where $\Dec$ can successfully retrieve the keys and decrypt the state, based on only the output ciphertext and the secret key.

\subsubsection{Security of verification.}

The trap code is proven secure in its application to one-time programs~\cite{BGS13}. Broadbent and Wainewright proved authentication security (with an explicit, efficient simulator)~\cite{BW16}. One can use similar strategies to prove $\kappa$-IND-VER for $\TC$. In fact, $\TC$ satisfies a stronger notion of verifiability, where the adversary is allowed to submit plaintexts in multiple rounds, which are either all encrypted or all swapped out. Two rounds are sufficient for us; the definitions and proof (see Appendix \ref{app:TC-defs} and Appendix \ref{app:TC-proof}) extend straightforwardly to the general case.

\begin{definition}[IND-VER-2 game]\label{def:VER-2-game}
	For an adversary $\advA = (\advA_0, \advA_1, \advA_2, \advA_3)$, a scheme $S$, and a security parameter $\kappa$, $\VerTwoGame{S}$ is shown in Figure~\ref{fig:ver2game}.
	
	\begin{figure}
		\centering
		\begin{tikzpicture}[scale=0.7, every node/.style={scale=0.7}]
		\draw (-4,0) rectangle (-3,3);
		\node[rotate=90] at (-3.5,1.5) {$S.\KeyGen(1^{\kappa})$};
		\draw (-3,2) -- (-2,2);
		\draw (-3,0.1) -- (-2.8,0.1);
		\draw (-3,0.2) -- (-2.8,0.2);
		\node[anchor=south] at (-2.5,2) {$\rho_{evk}$};
		\node[anchor=west] at (-2.8,0.2) {$sk$};
		
		\node[anchor=east] at (-3,4.5) {$|0^{n_2}\rangle\langle0^{n_2}|$};
		\node[anchor=east] at (-3,5.5) {$|0^{n_1}\rangle\langle0^{n_1}|$};
		\node[anchor=east] at (-3,6.5) {$r \in_R \{0,1\}$};
		\draw (-3,6.5) -- (11.5,6.5);
		\draw (-3,6.6) -- (11.5,6.6);
		\draw (-3,5.5) -- (7.5,5.5);
		\draw (-3,4.5) -- (7.5,4.5);
		
		\node at (-1.5,1.875) {$\advA_0$};
		\draw (-2,0) rectangle (-1,3.75);
		\draw (-1,0.5) -- (2,0.5);
		\draw (-1,3.5) -- (2,3.5);
		\node[anchor=south west] at (-1,0.5) {$R$};
		\node[anchor=south west] at (-1,3.5) {$X_1$};
		
		\node at (2.5,1.875) {$\advA_1$};
		\draw (2,0) rectangle (3,3.75);
		\draw (3,0.5) -- (6,0.5);
		\draw (3,3.5) -- (6,3.5);
		\node[anchor=south west] at (3,0.5) {$R'$};
		\node[anchor=south west] at (3,3.5) {$X_2$};
		
		\filldraw[fill=white] (0.3,3.25) rectangle (1.7,3.75);
		\node at (1,3.5) {$S.\Enc_{sk}$};
		\node at (0,6.53) {$\bullet$};
		\draw (0,6.5) -- (0,3.5);
		\node at (0,5.5) {$\times$};
		\node at (0,3.5) {$\times$};

		\filldraw[fill=white] (4.3,3.25) rectangle (5.7,3.75);
		\node at (5,3.5) {$S.\Enc_{sk}$};
		\node at (4,6.53) {$\bullet$};
		\draw (4,6.5) -- (4,3.5);
		\node at (4,4.5) {$\times$};
		\node at (4,3.5) {$\times$};
		
		\filldraw[fill=white] (6,0) rectangle (7,3.75);
		\node at (6.5,1.875) {$\advA_2$};
		\draw (7,0.5) -- (11,0.5);
		\draw (7,1.5) -- (8,1.5);
		\draw (7,1.6) -- (8,1.6);
		\draw (7,2.5) -- (8,2.5);
		\draw (7,2.6) -- (7.7,2.6) -- (7.7,4);
		\draw (7.8,4) -- (7.8,2.6) -- (8,2.6);
		\draw (7,3.5) -- (7.7,3.5);
		\draw (7.8,3.5) -- (8,3.5);
		\node[anchor=south west] at (6.9,3.5) {$C_{X'}$};
		\node[anchor=south west] at (7,2.55) {$c$};
		\node[anchor=south west] at (7,1.55) {$log$};
		\node[anchor=south west] at (7,0.5) {$R''$};
		
		\filldraw[fill=white] (7.5,4) rectangle (8.5,6);
		\node at (8,5) {$\Phi_c$};
		\draw (8.5,5) -- (11.5,5);

		\draw (8,1) rectangle (9,3.75);
		\node[rotate=90] at (8.5,2.375) {$S.\VerDec_{sk}$};
		\draw (9,1.5) -- (11,1.5);
		\draw (9,1.6) -- (11,1.6);
		\draw (9,2.5) -- (11,2.5);
		\draw (9,2.6) -- (11,2.6);
		\draw (9,3.5) -- (11,3.5);
		\node[anchor=north] at (10,1.6) {$acc(0)/rej(1)$};
		\node[anchor=south west] at (9,2.55) {$c$};
		\node[anchor=south west] at (9,3.5) {$X'$};
		
		\node at (9.75,1.55) {$\bullet$};
		\draw (9.75,1.6) -- (9.75,4.5);
		\filldraw[fill=white] (9.25,4.5) rectangle (10.25,5.5);
		\node at (9.75,5) {\faBan};
		\node at (10.5,6.53) {$\bullet$};
		\draw (10.5,6.5) -- (10.5,3.5);
		\node at (10.5,5) {$\times$};
		\node at (10.5,3.5) {$\times$};
		
		\draw (11,0) rectangle (12,3.75);
		\node at (11.5,1.875) {$\advA_3$};
		\draw (12,1.875) -- (12.5,1.875);
		\draw (12,1.975) -- (12.5,1.975);
		\node[anchor=west] at (12.5,1.975) {$r'$};
		\end{tikzpicture}
		\caption{The game $\VerTwoGame{S}$.}\label{fig:ver2game}
	\end{figure}
\end{definition}

\begin{definition}[$\kappa$-IND-VER-2]
	A vQFHE scheme $S$ satisfies \emph{$\kappa$-IND-VER-2} if for any QPT adversary $\advA$,
$
	\Pr[\VerTwoGame{S} = 1] \leq \frac{1}{2} + \negl(\kappa).
$
\end{definition}

\begin{theorem}\label{thm:tc-security-2}
	$\TC$ is $\kappa$-IND-VER-2 for the above circuit class.
\end{theorem}

\section{\texorpdfstring{$\TTP$}{TrapTP}: Quantum FHE With Verification}\label{sec:ttp-definition}

In this section, we introduce our candidate scheme for verifiable quantum fully homomorphic encryption (vQFHE). In this section, we will define the scheme prove correctness, compactness, and privacy. We will show verifiability in Section~\ref{sec:ttp-verifiability}.

Let $\kappa \in \mathbb{N}$ be a security parameter, and let $t, p, h \in \mathbb{N}$ be an upper bound on the number of $\T$, $\P$, and $\H$ gates (respectively) that will be in the circuit which is to be homomorphically evaluated. \todo{Mention something about layered schemes?} As in Section~\ref{sec:trap-code-encrypt}, we fix a self-dual $[[m,1,d]]$ CSS code $\CSS$ which has $m = \poly(d)$ and can correct $d_c := \kappa$ errors (e.g., the concatenated Steane code). We also fix a classical fully homomorphic public-key encryption scheme $\HE$ with decryption in $\mathsf{LOGSPACE}$ (see, e.g., \cite{BV11}). Finally, fix a message authentication code $\MAC = (\Tag, \Ver)$ that is existentially unforgeable under adaptive chosen message attacks (EUF-CMA~\cite{KL14}) from a quantum adversary; for example, one may take the standard pseudorandom-function construction with a post-quantum PRF. This defines an authentication procedure $\MAC.\Sign_k: m \mapsto (m, \MAC.\Tag_k(m))$.

\subsubsection{Key generation and encryption.}

The evaluation key will require a number of auxiliary states, which makes the key generation algorithm $\TTP.\KeyGen$ somewhat involved (see Algorithm~\ref{alg:ttp-key-gen} and Algorithm~\ref{alg:ttp-gadget-gen}). Note that non-evaluation keys are generated first, and then used to encrypt auxiliary states which are included in the evaluation key (see $\TTP.\Enc$ below). Most states are encrypted using the same `global' permutation $\pi$, but all qubits in the error-correction gadget (except first and last) are encrypted using independent permutations $\pi_i$ (see line~\ref{line:gadget-gen}). The $\T$-gate gadgets are prepared by Algorithm \ref{alg:ttp-gadget-gen}, making use of garden-hose gadgets from~\cite{DSS16}. \arxiv{The structure of these gadgets is described in Figure~\ref{fig:ttp-gadget}.}{}

\begin{algorithm}
	\caption{$\TTP.\KeyGen(1^\kappa, 1^t, 1^p, 1^h)$}\label{alg:ttp-key-gen}
	\begin{algorithmic}[1]
		\State $k \gets \MAC.\KeyGen(1^{\kappa})$
		\State $\pi \gets_R S_{3m}$\Comment{$S_{3m}$ is the permutation group on $3m$ elements}
		\For {i = 0, ..., t}
			\State $(sk_i, pk_i, evk_i) \gets \HE.\KeyGen(1^\kappa)$
		\EndFor
		\State $sk \gets (\pi, k, sk_0, ..., sk_t, pk_0)$
		\For {i = 1, ..., p}\Comment{Magic-state generation for $\P$}
			\State $\mu^{\P}_i \gets \TTP.\Enc(sk, \P\ket{+})$\Comment{See Algorithm~\ref{alg:ttp-enc} for $\TTP.\Enc$}
		\EndFor
		\For {i = 1, ..., t}\Comment{Magic-state generation for $\T$}
		\State $\mu^{\T}_i \gets \TTP.\Enc(sk, \T\ket{+})$
		\EndFor
		\For {i = 1, ..., h}\Comment{Magic-state generation for $\H$}
		\State $\mu^{\H}_i \gets \TTP.\Enc(sk, \frac{1}{\sqrt{2}}(\H \otimes \id)(\ket{00}+\ket{11}))$
		\EndFor
		\For {i = 1, ..., t}\Comment{Gadget generation for $\T$}
			\State $\pi_i \gets_R S_{3m}$
			\State $(g_i, \gin_i, \gmid_i, \gout_i) \gets \TTP.\GadgetGen(sk_{i-1})$\Comment{See Algorithm~\ref{alg:ttp-gadget-gen}}
			\State $\Gamma_i \gets \MAC.\Sign(\HE.\Enc_{pk_i}(g_i, \pi_i)) \otimes \TTP.\Enc((\pi_i, k, sk_0, ..., sk_t, pk_i), \gmid_i) \otimes \TTP.\Enc(sk, \gin_i, \gout_i)$\label{line:gadget-gen}
		\EndFor
		\State $\mathit{keys} \gets \MAC.\Sign(evk_0, ..., evk_t, pk_0, ..., pk_t, \HE.\Enc_{pk_0}(\pi))$
		\State $\rho_{evk} \gets (\mathit{keys}, \mu^{\P}_0, ..., \mu^{\P}_p, \mu^{\T}_0, ..., \mu^{\T}_t, \mu^{\H}_0, ..., \mu^{\H}_h, \Gamma_1, ..., \Gamma_t)$
		\State \textbf{return} $(sk, \rho_{evk})$
	\end{algorithmic}
\end{algorithm}

\begin{algorithm}
	\caption{$\TTP.\GadgetGen(sk_i)$}\label{alg:ttp-gadget-gen}
	\begin{algorithmic}[1]
		\State $g_i \gets g(sk_i)$\Comment{classical description of the garden-hose gadget, see \cite{DSS16}, p. 13}
		\State $(\gin, \gmid, \gout) \gets$ \arxiv{generate $\ket{\Phi^+}$ states depending on $g_i$ as in Figure~\ref{fig:ttp-gadget}.}{generate $\ket{\Phi^+}$ states and arrange them as described by $g_i$. Call the first qubit $\gin_i$ and the last qubit $\gout_i$. The rest forms the state $\gmid_i$.}
		\State \textbf{return} $(g_i, \gin_i, \gmid_i, \gout_i)$
	\end{algorithmic}
\end{algorithm}

\arxiv{
\begin{figure}
\centering
\begin{tikzpicture}
\foreach \x in {0,...,5}
	\filldraw[fill=black] (\x,0) circle (2pt);
\draw[decorate,decoration={snake,amplitude=1pt,segment length=5pt}] (0,0) to[bend right] (2,0);
\draw[decorate,decoration={snake,amplitude=1pt,segment length=5pt}] (1,0) to[bend right] (5,0);
\draw[decorate,decoration={snake,amplitude=1pt,segment length=5pt}] (3,0) to[bend right] (4,0);
\node at (0,0.75) {$\gin_{i}$};
\node at (2.5,0.75) {$\gmid_{i}$};
\node at (5,0.75) {$\gout_{i}$};
\filldraw[fill=white] (3.3,-0.3) rectangle (3.7,0.1);
\node at (3.5,-0.1) {$\P$};
\draw[decoration={brace,raise=8pt,amplitude=5pt}, decorate] (0.75,0) -- (4.25,0);
\end{tikzpicture}
\caption{A garden-hose gadget consists of a number of EPR pairs, arranged in a specific order (described by the classical string$g_i$). The total number of EPR pairs depends on the \emph{garden-hose complexity} of the function $\HE.\Dec$~\cite{BFSS13, DSS16}. On one of the EPR pairs, a $\P$ gate is applied (it is of the form $(\P \otimes \id)\ket{\Phi^+}$). To use the gadget, the evaluator ``teleports in'' a data qubit by performing a Bell measurement between that qubit and $\gin_i$. Then, several Bell measurements are performed on $\gmid_i$, causing the data qubit to either pass the $\P$ gate, or not. Which measurements are performed, depend on classical information held by the evaluator (see also Algorithm~\ref{alg:ttp-cond-p}). The data, possibly with a $\P$ gate applied to it, ends up at $\gout_i$.}\label{fig:ttp-gadget}
\end{figure}
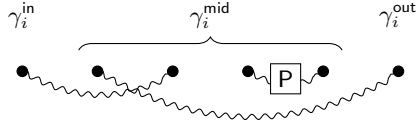
}{}

The encryption of a quantum state is similar to $\TC.\Enc$, only the keys to the QOTP are now chosen during encryption (rather than during key generation) and appended in encrypted and authenticated form to the ciphertext (see Algorithm~\ref{alg:ttp-enc}). Note that the classical secret keys $sk_0$ through $sk_t$ are not used. 

\begin{algorithm}
	\caption{$\TTP.\Enc((\pi, k, sk_0, ..., sk_t, pk), \sigma)$}\label{alg:ttp-enc}
	\begin{algorithmic}[1]
		\State $\enc{\sigma} \gets \displaystyle \sum_{x,z \in \{0,1\}^{3m}} \Big(\TC.\Enc((\pi,x,z),\sigma) \otimes \MAC.\Sign_k(\HE.\Enc_{pk}(x,z))\Big)$\Comment{Algorithm~\ref{alg:tc-enc}}\label{line:ttp-enc}
		\State \textbf{return} $\enc{\sigma}$
	\end{algorithmic}
\end{algorithm}

\subsubsection{Evaluation.}\label{sec:ttp-evaluation}

Evaluation of gates is analogous to the interactive evaluation scheme using the trap code~\cite{BGS13},  except the interactions are replaced by classical homomorphic evaluation. Evaluation of a circuit $c$ is done gate-by-gate, as follows.

In general, we will use the notation $\enc{\cdot}$ to denote encryptions of classical and quantum states. For example, in the algorithms below, $\enc{\sigma}$ is the encrypted input qubit for a gate and $\enc{x}$ and $\enc{z}$ are classical encryptions of the associated QOTP keys. We will assume throughout that $\HE.\Enc$ and $\HE.\Eval$ produce, apart from their actual output, a complete \emph{computation log} describing all randomness used, all computation steps, and all intermediate results. 

\paragraph{Measurements.}\label{sec:ttp-eval-meas}

Computational basis measurement is described in Algorithm~\ref{alg:ttp-measure}. Recall that $\TC.\VerDecMeasurement$ is a completely classical procedure that decodes the list of $3m$ measurement outcomes into the logical outcome and checks the relevant traps. Hadamard-basis measurement is performed similarly, except the qubits of $\enc{\sigma}$ are measured in the Hadamard basis and $\HE.\Enc_{pk}(\times)$ is given as the last argument for the evaluation of $\TC.\VerDecMeasurement$.

\begin{algorithm}
	\caption{$\TTP.\Eval\Measure(\enc{\sigma}, \enc{x}, \enc{z},\enc{\pi},pk,evk)$}\label{alg:ttp-measure}
	\begin{algorithmic}[1]
		\State $a = (a_1, ..., a_{3m}) \gets $ measure qubits of $\enc{\sigma}$ in the computational basis
		\State $(\enc{a},log_1) \gets \HE.\Enc_{pk}(a)$
		\State $(\enc{b}, \enc{\flag},log_2) \gets \HE.\Eval_{evk}^{\TC.\VerDecMeasurement}((\enc{\pi},\enc{x},\enc{z}),\enc{a},\HE.\Enc_{pk}(+))$
		\State \textbf{return} $(\enc{b},\enc{\flag},log_1,log_2)$\Comment{$b \in \{0,1\}$ represents the output of the measurement}
	\end{algorithmic}
\end{algorithm}

\paragraph{Pauli gates.}\label{sec:ttp-eval-pauli}
A logical Pauli-$\X$ is performed by (homomorphically) flipping the $\X$-key bits of the QOTP (see Algorithm~\ref{alg:ttp-eval-x}). Since this is a classical operation, the functionality extends straightforwardly to a classically controlled Pauli-$\X$ (by specifying an additional bit $b$ encrypted into $\enc{b}$ that indicates whether or not $\X$ should be applied; see Algorithm~\ref{alg:ttp-eval-conditional-x}). The (classically controlled) evaluation of a Pauli-$\Z$ works the same way, only the relevant bits in $\enc{z}$ are flipped.

\begin{algorithm}
	\caption{$\TTP.\Eval\X(\enc{\sigma}, \enc{x},\enc{\pi},pk,evk)$}\label{alg:ttp-eval-x}
	\begin{algorithmic}[1]
		\State $(\enc{x},log_1) \gets \HE.\Eval_{evk}^{\mathsf{unpermute}}(\enc{\pi},\enc{x})$
		\State $(\enc{x},log_2) \gets \HE.\Eval_{\evk}^{\oplus}(\enc{x}, \HE.\Enc_{pk}(1^m0^{2m}))$\Comment{this flips the first $m$ bits}
		\State $(\enc{x},log_3) \gets \HE.\Eval_{evk}^{\mathsf{permute}}(\enc{\pi},\enc{x})$
		\State \textbf{return} $(\enc{\sigma},\enc{x},log_1, log_2, log_3)$
	\end{algorithmic}
\end{algorithm}

\begin{algorithm}
	\caption{$\TTP.\Eval\Conditional\X(\enc{b}, \enc{\sigma}, \enc{x}, \enc{z}, \enc{\pi},pk,evk)$}\label{alg:ttp-eval-conditional-x}
	\begin{algorithmic}[1]
		\State $(\enc{x},log_1) \gets \HE.\Eval_{evk}^{\mathsf{unpermute}}(\enc{\pi},\enc{x})$
		\State $\enc{s} \gets \HE.\Eval_{evk}^{y \mapsto y^m0^{2m}}(\enc{b})$
		\State $(\enc{x},log_2) \gets \HE.\Eval_{evk}^{\oplus}(\enc{x},\enc{s})$\Comment{this conditionally flips the first $m$ bits}
		\State $(\enc{x},log_3) \gets \HE.\Eval_{evk}^{\mathsf{permute}}(\enc{\pi},\enc{x})$
		\State \textbf{return} $(\enc{\sigma},\enc{x}, \enc{z},log_1, log_2, log_3)$
	\end{algorithmic}
\end{algorithm}

\paragraph{\texorpdfstring{$\CNOT$}{CNOT} gates.}\label{sec:ttp-eval-cnot}

The evaluation of $\CNOT$ in $\TTP$ is analogous to $\TC$, only the key updates are performed homomorphically during evaluation (see Algorithm~\ref{alg:ttp-eval-cnot}).

\begin{algorithm}
	\caption{$\TTP.\Eval\CNOT(\enc{\sigma_1}, \enc{\sigma_2}, \enc{x_1},\enc{x_2}, \enc{z_1}, \enc{z_2}, \enc{\pi},pk,evk)$}\label{alg:ttp-eval-cnot}
	\begin{algorithmic}[1]
		\State $(\enc{\sigma_1},\enc{\sigma_2}) \gets$ apply $\CNOT$ on all physical qubit pairs of $\enc{\sigma_1}, \enc{\sigma_2}$
		\State $(\enc{x_1}, \enc{x_2}, \enc{z_1}, \enc{z_2},log_1) \gets \HE.\Eval_{evk}^{\CNOT\mathsf{-key-update}}(\enc{x_1},\enc{x_2},\enc{z_1},\enc{z_2})$\Comment{for commutation rules, see Section~\ref{sec:tc-eval-cnot}}
		\State \textbf{return} $(\enc{\sigma_1}, \enc{\sigma_2},\enc{x_1},\enc{x_2}, \enc{z_1}, \enc{z_2}, log_1, log_2)$
	\end{algorithmic}
\end{algorithm}

\paragraph{Phase gates.}

Performing a $\P$ gate requires homomorphic evaluation of all the above gates: (classically controlled) Paulis, $\CNOT$s, and measurements. We also consume the state $\mu^{\P}_i$ (an encryption of the state $\P\ket{+}$) for the $i$th phase gate in the circuit. The circuit below applies $\P$ to the data qubit (see, e.g., \cite{BGS13}). 
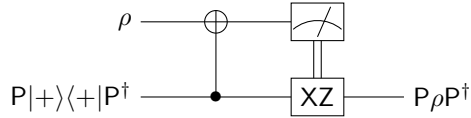
\begin{center}
	\begin{tikzpicture}
	\node[anchor=east] at (0,1) {$\rho$};
	\node[anchor=east] at (0,0) {$\P\ket{+}\bra{+}\P^{\dagger}$};
	\draw (0,0) -- (3.5,0);
	\draw (0,1) -- (2,1);
	\node at (1,0) {$\bullet$};
	\draw (1,1) circle (0.15);
	\draw (1,0) -- (1,1.15);
	\draw (2.3,1) -- (2.3,0);
	\draw (2.4,1) -- (2.4,0);
	\node at (2.35,1) {\meas};
	\filldraw[fill=white] (2,-0.25) rectangle (2.7,0.25);
	\node at (2.35,0) {$\X\Z$};
	\node[anchor=west] at (3.5,0) {$\P\rho\P^{\dagger}$};
	\end{tikzpicture}
\end{center}
We define $\TTP.\Eval\P$ to be the concatenation of the corresponding gate evaluations. The overall computation log is just a concatenation of the logs.

\paragraph{Hadamard gate.}

The Hadamard gate can be performed using the same ingredients as the phase gate~\cite{BGS13}. The $i$th gate consumes $\mu^{\H}_i$, an encryption of $(\H \otimes \id)\ket{\Phi^+}$.
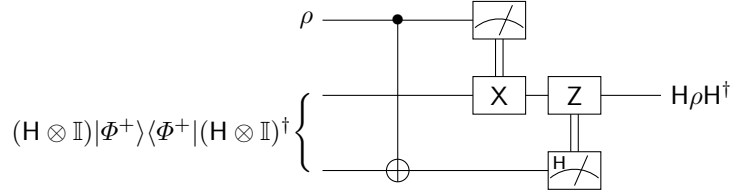
\begin{center}
	\begin{tikzpicture}
	\node[anchor=east] at (0,2) {$\rho$};
	\node[anchor=east] at (0,0.5) {$(\H\otimes\id)\ket{\Phi^+}\bra{\Phi^+}(\H\otimes\id)^{\dagger}\Bigg\{$};
	\draw (0,0) -- (3.5,0);
	\draw (0,1) -- (4.5,1);
	\draw (0,2) -- (2.5,2);
	\node at (1,2) {$\bullet$};
	\draw (1,0) circle (0.15);
	\draw (1,-0.15) -- (1,2);
	\draw (2.3,1) -- (2.3,2);
	\draw (2.4,1) -- (2.4,2);
	\node at (2.35,2) {\meas};
	\draw (3.3,1) -- (3.3,0);
	\draw (3.4,1) -- (3.4,0);
	\node at (3.35,0) {\meas};
	\node at (3.15,0.1) {$\scriptstyle{\mathsf{H}}$};
	\filldraw[fill=white] (3,0.75) rectangle (3.7,1.25);
	\node at (3.35,1) {$\Z$};
	\filldraw[fill=white] (2,0.75) rectangle (2.7,1.25);
	\node at (2.35,1) {$\X$};
	\node[anchor=west] at (4.5,1) {$\H\rho\H^{\dagger}$};
	\end{tikzpicture}
\end{center}

\paragraph{The \texorpdfstring{$\T$}{T} gate.}

A magic-state computation of $\T$ uses a similar circuit to that for $\P$, using $\mu^{\T}_i$ (an encryption of $\T\ket{+}$) as a resource for the $i$th $\T$ gate:
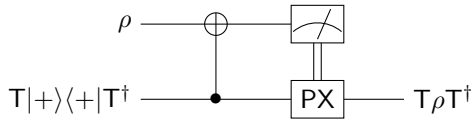
\begin{center}
	\begin{tikzpicture}
	\node[anchor=east] at (0,1) {$\rho$};
	\node[anchor=east] at (0,0) {$\T\ket{+}\bra{+}\T^{\dagger}$};
	\draw (0,0) -- (3.5,0);
	\draw (0,1) -- (2,1);
	\node at (1,0) {$\bullet$};
	\draw (1,1) circle (0.15);
	\draw (1,0) -- (1,1.15);
	\draw (2.3,1) -- (2.3,0);
	\draw (2.4,1) -- (2.4,0);
	\node at (2.35,1) {\meas};
	\filldraw[fill=white] (2,-0.25) rectangle (2.7,0.25);
	\node at (2.35,0) {$\P\X$};
	\node[anchor=west] at (3.5,0) {$\T\rho\T^{\dagger}$};
	\end{tikzpicture}
\end{center}
The evaluation of this circuit is much more complicated, since it requires the application of a classically-controlled phase correction $\P$. We will accomplish this using the error-correction gadget $\Gamma_i$.

First, we remark on some subtleties regarding the encrypted classical information surrounding the gadget. Since the structure of $\Gamma_i$ depends on the classical secret key $sk_{i-1}$, the classical information about $\Gamma_i$ is encrypted under the (independent) public key $pk_i$ (see Algorithm~\ref{alg:ttp-key-gen}). This observation will play a crucial role in our proof that $\TTP$ satisfies IND-VER, in Section~\ref{sec:ttp-verifiability}.

The usage of two different key sets also means that, at some point during the evaluation of a $\T$ gate, all classically encrypted information needs to be recrypted from the $(i-1)$st into the $i$th key set. This can be done because $\enc{sk}_{i-1}$ is included in the classical information $g_i$ in $\Gamma_i$. The recryption is performed right before the classically-controlled phase gate is applied (see Algorithm~\ref{alg:ttp-eval-t}).

\begin{algorithm}
	\caption{$\TTP.\Eval\T(\enc{\sigma}, \enc{x}, \enc{z}, \enc{\pi}, \mu_i^{\T}, \Gamma_i, pk_{i-1}, evk_{i-1}, pk_i, evk_i)$}\label{alg:ttp-eval-t}
	\begin{algorithmic}[1]
		\State $(\enc{\sigma_1}, \enc{\sigma_2}, \enc{x_1}, \enc{z_1}, \enc{x_2}, \enc{z_2}, log_1) \gets \TTP.\Eval\CNOT(\mu_i^{\T}, \enc{\sigma}, \enc{x}, \enc{z}, \enc{\pi}, pk_{i-1}, evk_{i-1})$
		\State $(\enc{b},log_2) \gets \TTP.\Eval\Measure(\enc{\sigma_2}, \enc{x_2}, \enc{z_2}, \enc{\pi}, pk_{i-1}, evk_{i-1})$
		\State $log_3 \gets$ recrypt \emph{all} classically encrypted information (except $\enc{b}$) from key set $i-1$ into key set $i$.
		\State $(\enc{\sigma},log_4) \gets \TTP.\Eval\Conditional\P(\enc{b},\enc{\sigma_1}, \enc{x_1}, \enc{z_1}, \Gamma_i,\enc{\pi},pk_i,evk_i)$
		\State \textbf{return} $(\enc{\sigma}, log_1, log_2, log_3, log_4)$
	\end{algorithmic}
\end{algorithm}

Algorithm~\ref{alg:ttp-cond-p} shows how to use $\Gamma_i$ to apply logical $\P$ on an encrypted quantum state $\enc{\sigma}$, conditioned on a classical bit $b$ for which only the encryption $\enc{b}$ is available. When $\TTP.\Eval\Conditional\P$ is called, $b$ is encrypted under the $(i-1)$st classical $\HE$-key, while all other classical information (QOTP keys $x$ and $z$, permutations $\pi$ and $\pi_i$, classical gadget description $g_i$) is encrypted under the $i$th key. Note that we can evaluate Bell measurements using only evaluation of $\CNOT$, computational-basis measurements, and $\H$-basis measurements. In particular, no magic states are needed to perform a Bell measurement.
After this procedure, the data is in qubit $\enc{\gout_i}$. The outcomes $a_1, a_2, a$ of the Bell measurements determine how the keys to the QOTP must be updated.

\begin{algorithm}
	\caption{${\TTP.\Eval\Conditional\P(\enc{b},\enc{\sigma}, \enc{x}, \enc{z}, \Gamma_i  = (\enc{g_i}, \enc{\pi_i}, \enc{\gin_i}, \enc{\gmid_i}, \enc{\gout_i}), \enc{\pi}, pk_i, evk_i)}$}\label{alg:ttp-cond-p}
	\begin{algorithmic}[1]
		\State $(\enc{a_1}, \enc{a_2}, log_1) \gets$ evaluate Bell measurement between $\enc{\sigma}$ and $\enc{\gin_i}$\Comment{$a_1,a_2 \in \{0,1\}$}
		\State $(\enc{a},log_2) \gets$ evaluate Bell measurements in $\enc{\gmid_i}$ as dictated by the ciphertext $\enc{b}$ and the garden-hose protocol for $\HE.\Dec$
		\State $(\enc{x},\enc{z}, log_3) \gets \HE.\Eval_{evk_{i}}^{\mathsf{T-key-update}}(\enc{x},\enc{z},\enc{a_1},\enc{a_2},\enc{a}, \enc{g_i})$
		\State \textbf{return} $(\enc{\gout_i}, \enc{x}, \enc{z}, log_1, log_2, log_3)$
	\end{algorithmic}
\end{algorithm}

\subsubsection{Verified Decryption.}

The decryption procedure (Algorithm~\ref{alg:ttp-verdec}) consists of two parts. First, we perform several classical checks. This includes MAC-verification of all classically authenticated messages, and checking that the gates listed in the log match the circuit description. We also check the portions of the log which specify the (purely classical, FHE) steps taken during $\HE.\Enc$ and $\HE.\Eval$; this is the standard transcript-checking procedure for FHE, which we call $\TTP.\CheckLog$. Secondly, we check all unmeasured traps and decode the remaining qubits. We reject if $\TTP.\CheckLog$ rejects, or if the traps have been triggered.
\begin{algorithm}
	\caption{$\TTP.\VerDec(sk, \enc{\sigma}, (\enc{x[i]})_i, (\enc{z[i]})_i, log, c)$}\label{alg:ttp-verdec}
	\begin{algorithmic}[1]
		\State Verify classically authenticated messages (in $log$) using $k$ (contained in $sk$). If one of these verifications rejects, \textbf{reject}.\label{line:mac-check}
		\State Check whether all claimed gates in $log$ match the structure of $c$. If not,  \textbf{return} $(\Omega, \rejstatepure)$.\Comment{Recall that $\Omega$ is a dummy state.}
		\State $\flag \gets \TTP.\CheckLog(log)$ If $\flag = $ rej, \textbf{return} $(\Omega, \rejstatepure)$.
		\State Check whether the claimed final QOTP keys in the $log$ match $\enc{x}$ and $\enc{z}$. If not, \textbf{return} $(\Omega, \rejstatepure)$.
		\ForAll{gates $G$ of $c$}
			\If{$G$ is a measurement}
				\State $\enc{x'},\enc{z'} \gets $ encrypted QOTP keys right before measurement (listed in $log$)
				\State $\enc{w} \gets$ encrypted measurement outcomes (listed in $log$)
				\State $x', z', w \gets \HE.\Dec_{sk_t}(\enc{x'},\enc{z'},\enc{w})$\label{line:decrypt1}
				\State Execute $\TC.\VerDecMeasurement((\pi, x', z'),w,\mathit{basis})$, where $\mathit{basis}$ is the appropriate basis for the measurement, and store the (classical) outcome.
				\If{a trap is triggered}
					\State  \textbf{return} $(\Omega, \rejstatepure)$. \label{line:ttp-verdec-end-of-ver}
				\EndIf
			\EndIf
		\EndFor
		\ForAll{unmeasured qubits $\enc{\sigma_i}$ in $\enc{\sigma}$}\label{line:decryptqubit}
			\State $x[i],z[i] \gets \HE.\Dec_{sk_t}(\enc{x[i]},\enc{z[i]})$\label{line:decrypt2}
			\State $\sigma_i \gets \TC.\VerDec_{(\pi, x[i], z[i])}(\enc{\sigma_i})$. If $\TC.\VerDec$ rejects, \textbf{return} $(\Omega,\rejstatepure)$.
		\EndFor
		\State $\sigma \gets $ the list of decrypted qubits (and measurement outcomes) that are part of the output of $c$
		\State \textbf{return} $(\sigma, \accstatepure)$\label{line:ttp-verdec-end-of-dec}
	\end{algorithmic}
\end{algorithm}

\subsection{Correctness, compactness, and privacy}

If all classical computation was unencrypted, checking correctness of $\TTP$ can be done by inspecting the evaluation procedure for the different types of gates, and comparing them to the trap code construction in~\cite{BGS13}. This suffices, since $\HE$ and the $\MAC$ authentication both satisfy correctness.

Compactness as defined in Definition~\ref{def:compactness-vqfhe} is also satisfied: verifying the computation log and checking all intermediate measurements (up until line~\ref{line:ttp-verdec-end-of-ver} in Algorithm~\ref{alg:ttp-verdec}) is a completely classical procedure and runs in polynomial time in its input. The rest of $\TTP.\VerDec$ (starting from line~\ref{line:decryptqubit}) only uses the secret key and the ciphertext ($\enc{\sigma}, \enc{x}, \enc{z})$ as input, not the log or the circuit description. Thus, we can separate $\TTP.\VerDec$ into two algorithms $\Ver$ and $\Dec$ as described in Definition~\ref{def:compactness-vqfhe}, by letting the second part ($\Dec$, lines~\ref{line:decryptqubit} to~\ref{line:ttp-verdec-end-of-dec}) reject whenever the first part ($\Ver$, lines~\ref{line:mac-check} to~\ref{line:ttp-verdec-end-of-ver}) does. It is worth noting that, because the key-update steps are performed homomorphically during the evaluation phase, skipping the classical verification step yields a QFHE scheme without verification that satisfies Definition~\ref{def:compactness-qfhe} (and is authenticating). This is not the case for the scheme $\TC$, where the classical computation is necessary for the correct decryption of the output state.

In terms of privacy, $\TTP$ satisfies IND-CPA (see Section~\ref{sec:def-privacy}). This is shown by reduction to IND-CPA of $\HE$. This is non-trivial since the structure of the error-correction gadgets depends on the classical secret key. The reduction is done in steps, where first the security of the encryptions under $pk_{t}$ is applied (no gadget depends on $sk_t$), after which the quantum part of the gadget $\Gamma_t$ (which depends on $sk_{t-1}$) looks completely mixed from the point of view of the adversary. We then apply indistinguishability of the classical encryptions under $pk_{t-1}$, and repeat the process. After all classical encryptions of the quantum one-time pad keys are removed, the encryption of a state appears fully mixed. Full details of this proof can be found in Lemma~1 of~\cite{DSS16}, where IND-CPA security of an encryption function very similar to $\TTP.\Enc$ is proven.

\section{Proof of verifiability for \texorpdfstring{$\TTP$}{TrapTP}}\label{sec:ttp-verifiability}
In this section, we will prove that $\TTP$ is $\kappa$-IND-VER. By Theorem~\ref{thm:IND-SEM}, it then follows that $\TTP$ is also verifiable in the semantic sense. We will define a slight variation on the VER indistinguishability game, followed by several hybrid schemes (variations of the $\TTP$ scheme) that fit into this new game. We will argue that for any adversary, changing the game or scheme does not significantly affect the winning probability. After polynomially-many such steps, we will have reduced the adversary to an adversary for the somewhat homomorphic scheme $\TC$, which we already know to be IND-VER. This will complete the argument that $\TTP$ is IND-VER. The IND-VER game is adjusted as follows.\label{sec:security-hybrid}

\begin{definition}[Hybrid game $\HybridGame{S}$]\label{def:hybrid-game}
	For an adversary $\advA = (\advA_1, \advA_2, \advA_3)$, a scheme $S$, and security parameter $\kappa$, $\HybridGame{S}$ is the game in Figure~\ref{fig:hybgame}.
	
	\begin{figure}
		\centering
		\makebox[\textwidth][c]{
		\begin{tikzpicture}[scale = 0.92]
		\draw (0,0) rectangle (1,3);
		\node[rotate=90] at (0.5,1.5) {$S.\KeyGen(1^{\kappa})$};
		\draw (1,2) -- (2,2);
		\draw (1,0.1) -- (1.2,0.1);
		\draw (1,0.2) -- (1.2,0.2);
		\node[anchor=south] at (1.5,2) {$\rho_{evk}$};
		\node[anchor=west] at (1.2,0.2) {$sk$};
		
		\node at (2.5,1.875) {$\advA_1$};
		\draw (2,0) rectangle (3,3.75);
		\draw (3,0.5) -- (6,0.5);
		\draw (3,3.5) -- (6,3.5);
		\node[anchor=south west] at (3,0.5) {$R$};
		\node[anchor=south west] at (3,3.5) {$X$};
		\node[anchor=east] at (3,5.5) {$|0^n\rangle\langle0^n|$};
		\node[anchor=east] at (3,6.5) {$r \in_R \{0,1\}$};
		\draw (3,6.5) -- (11.5,6.5);
		\draw (3,6.6) -- (11.5,6.6);
		\draw (3,5.5) -- (11.5,5.5);
		
		\filldraw[fill=white] (4.3,3.25) rectangle (5.7,3.75);
		\node at (5,3.5) {$S.\Enc_{sk}$};
		\node at (4,6.53) {$\bullet$};
		\draw (4,6.5) -- (4,3.5);
		\node at (4,5.5) {$\times$};
		\node at (4,3.5) {$\times$};
		
		\filldraw[fill=white] (6,0) rectangle (7,3.75);
		\node at (6.5,1.875) {$\advA_2$};
		\draw (7,0.5) -- (8.5,0.5); \draw (8.6,0.5) -- (11,0.5);
		\draw (7,1.5) -- (8,1.5);
		\draw (7,1.6) -- (8,1.6);
		\draw (7,2.5) -- (8,2.5);
		\draw (7,2.6) -- (7.7,2.6) -- (7.7,5);
		\draw (7.8,5) -- (7.8,2.6) -- (8,2.6);
		\draw (7,3.5) -- (7.7,3.5);
		\draw (7.8,3.5) -- (8,3.5);
		\node[anchor=south west] at (6.9,3.5) {$C_{X'}$};
		\node[anchor=south west] at (7,2.55) {$c$};
		\node[anchor=south west] at (7,1.55) {$log$};
		\node[anchor=south west] at (7,0.5) {$R'$};
		
		\filldraw[fill=white] (7.5,5) rectangle (8.5,6);
		\node at (8,5.5) {$\Phi_c$};

		\draw (8,1) rectangle (9,3.75);
		\node[rotate=90] at (8.5,2.375) {$S.\VerDec_{sk}$};
		\draw (9,1.5) -- (11,1.5);
		\draw (9,1.6) -- (11,1.6);
		\draw (9,2.5) -- (11,2.5);
		\draw (9,2.6) -- (11,2.6);
		\draw (9,3.5) -- (11,3.5);
		\node[anchor=north] at (10,1.6) {\small $acc/rej$};
		\node[anchor=south west] at (9,2.55) {$c$};
		\node[anchor=south west] at (9,3.5) {$X'$};
		
		\node at (9.75,1.55) {$\bullet$};
		\draw (9.75,1.6) -- (9.75,5.5);
		\filldraw[fill=white] (9.25,5) rectangle (10.25,6);
		\node at (9.75,5.5) {\faBan};
		\node at (10.5,6.53) {$\bullet$};
		\draw (10.5,6.5) -- (10.5,3.5);
		\node at (10.5,5.5) {$\times$};
		\node at (10.5,3.5) {$\times$};
		
		\draw (11,0) rectangle (12,3.75);
		\node at (11.5,1.875) {$\advA_3$};
		\draw (12,1.875) -- (12.5,1.875);
		\draw (12,1.975) -- (12.5,1.975);
		\node[anchor=west] at (12.5,1.975) {$r'$};
		
		\draw (5.5,3.75) -- (5.5,4.3) -- (7.7,4.3);
		\draw (7.8,4.3) -- (8.3,4.3) -- (8.3,3.75);
		\draw (5.6,3.75) -- (5.6,4.2) -- (7.7,4.2);
		\draw (7.8,4.2) -- (8.2,4.2) -- (8.2,3.75);
		
		\draw (0.9,0) -- (0.9,-0.5) -- (8.25,-0.5) -- (8.25,1);
		\draw (0.6,0) -- (0.6,-1) -- (8.6,-1) -- (8.6,1);
		\draw (0.7,0) -- (0.7,-0.9) -- (8.5,-0.9) -- (8.5,1);
		\end{tikzpicture}
	}
	\caption{The hybrid indistinguishability game $\HybridGame{S}$, which is a slight variation on $\VerGame{S}$ from Figure~\ref{fig:vergame}.}\label{fig:hybgame}
	\end{figure}
\end{definition}

Comparing to Definition~\ref{def:VER-game}, we see that three new wires are added: a classical wire from $S.\Enc$ to $S.\VerDec$, and a classical and quantum wire from $S.\KeyGen$ to $S.\VerDec$. We will later adjust $\TTP$ to use these wires to bypass the adversary; $\TTP$ as defined in the previous section does not use them. Therefore, for any adversary,
$\Pr[\VerGame{\TTP} = 1] = \Pr[\HybridGame{\TTP} = 1].$

\subsubsection{Hybrid 1: Removing Classical MAC.}

In $\TTP$, the initial keys to the QOTP can only become known to $\VerDec$ through the adversary. We thus use $\MAC$ to make sure these keys cannot be altered. Without this authentication, the adversary could, e.g., homomorphically use $\enc{\pi}$ to flip only those bits in $\enc{x}$ that correspond to non-trap qubits, thus applying $\X$ to the plaintext. In fact, all classical information in the evaluation key must be authenticated.

In the first hybrid, we argue that the winning probability of a QPT $\advA$ in $\HybridGame{\TTP}$ is at most negligibly higher than in $\HybridGame{\TTP'}$, where $\TTP'$ is a modified version of $\TTP$ where the initial keys are sent directly from $\KeyGen$ and $\Enc$ to $\VerDec$ (via the extra wires above). More precisely, in $\TTP'.\KeyGen$ and $\TTP'.\Enc$, whenever  $\MAC.\Sign(\HE.\Enc(x))$ or $\MAC.\Sign(x)$ is called, the message $x$ is also sent directly to $\TTP'.\VerDec$. Moreover, instead of decrypting the classically authenticated messages sent by the adversary, $\TTP'.\VerDec$ uses the information it received directly from $\TTP'.\KeyGen$ and $\TTP'.\Enc$. It still check whether the computation log provided by the adversary contains these values at the appropriate locations and whether the $\MAC$ signature is correct. The following fact is then a straightforward consequence of the EUF-CMA property of $\MAC$.

Recall that all adversaries are QPTs, i.e., quantum polynomial-time uniform algorithms. Given two hybrid games $H_1, H_2$, and a QPT adversary $\advA$, define
	\[
\AdvHyb{H_1}{H_2} := 
\bigl|\Pr[\HybridGame{\H_1} = 1] - \Pr[\HybridGame{\H_2} = 1]\bigr|\,.
	\]

\begin{lemma}\label{lem:remove-MAC}
	For any QPT $\advA$, $\AdvHyb{\TTP}{\TTP'} \leq \negl(\kappa)$.
\end{lemma}

\subsubsection{Hybrid 2: Removing Computation Log.}\label{sec:security-computation-log}

In $\TTP$ and $\TTP'$, the adversary (homomorphically) keeps track of the keys to the QOTP and stores encryptions of all intermediate values in the computation log. Whenever $\VerDec$ needs to know the value of a key (for example to check a trap or to decrypt the final output state), the relevant entry in the computation log is decrypted.

In $\TTP'$, however, the plaintext initial values to the computation log are available to $\VerDec$, as they are sent through the classical side channels. This means that whenever $\VerDec$ needs to know the value of a key, instead of decrypting an entry to the computation log, it can be computed by ``shadowing" the computation log in the clear.

For example, suppose the log contains the encryptions $\enc{b_1}, \enc{b_2}$ of two initial bits, and specifies the homomorphic evaluation of XOR, resulting in $\enc{b}$ where $b = b_1 \oplus b_2$. If one knows the plaintext values $b_1$ and $b_2$, then one can compute $b_1 \oplus b_2$ directly, instead of decrypting the entry $\enc{b}$ from the computation log.

We now define a second hybrid, $\TTP''$, which differs from $\TTP'$ exactly like this: $\VerDec$ still verifies the authenticated parts of the log, checks whether the computation log matches the structure of $c$, and checks whether it is syntactically correct. However, instead of decrypting values from the log (as it does in $\TTP.\VerDec$, Algorithm~\ref{alg:ttp-verdec}, on lines~\ref{line:decrypt1} and~\ref{line:decrypt2}), it computes those values from the plaintext initial values, by following the computation steps that are claimed in the log. By correctness of classical FHE, we then have the following.

\begin{lemma}\label{lem:remove-log}
	For any QPT $\advA$,
	$\AdvHyb{\TTP'}{\TTP''} \leq \negl(\kappa)$.
\end{lemma}
\begin{proof}
	Let $s$ be the (plaintext) classical information that forms the input to the classical computations performed by the adversary: initial QOTP keys, secret keys and permutations, measurement results, et cetera. Let $f$ be the function that the adversary computes on it in order to arrive at the final keys and logical measurement results.
	By correctness of $\HE$, we have that
	\[
	\Pr[\HE.\Dec_{sk_t}(\HE.\Eval^f_{evk_0, ..., evk_t}(\HE.\Enc_{pk_0}(s))) \neq f(s)] \leq \negl(\kappa).
	\]
	In the above expression, we slightly abuse notation and write $\HE.\Eval_{evk_0, ..., evk_t}$ to include the $t$ recryption steps that are performed during $\TTP.\Eval$. As long as the number of $\T$ gates, and thus the number of recryptions, is polynomial in $\kappa$, the expression holds.
	
	Thus, the probability that $\TTP'.\VerDec$ and $\TTP''.\VerDec$ use different classical values (decrypting from the log vs.\ computing from the initial values) is negligible. Since this is the only place where the two schemes differ, the output of the two $\VerDec$ functions will be identical, except with negligible probability. Thus $\advA$ will either win in both $\HybridGame{\TTP'}$ and $\HybridGame{\TTP''}$, or lose in both, again except with negligible probability.
	\qed
\end{proof}

\subsubsection{More Hybrids: Removing Gadgets.}

We continue by defining a sequence of hybrid schemes based on $\TTP''$. In $4t$ steps, we will move all error-correction functionality from the gadgets to $\VerDec$. This will imply that the adversary has no information about the classical secret keys (which are involved in constructing these gadgets). This will allow us to eventually reduce the security of $\TTP$ to that of $\TC$.

We remove the gadgets back-to-front, starting with the final gadget. Every gadget is removed in four steps. For all $1 \leq \ell \leq t$, define the hybrids $\TTP^{(\ell)}_1$, $\TTP^{(\ell)}_2$, $\TTP^{(\ell)}_3$, and $\TTP^{(\ell)}_4$ (with $\TTP_4^{(t+1)} := \TTP''$) as follows:

\todo[inline]{We could consider actually writing out the PostGadgetGen and FakeGadgetGen, and perhaps even the hybrid VerDec. It does cost a lot of space and time though, so the decision can go either way.
\\Y: I did write out the PostGadgetGen and FakeGadgetGen, but not VerDec (yet). I don't know if it adds much to do so.}

\arxiv{\begin{enumerate}}{}
\arxiv{\item}{\textbf{1}.} $\TTP^{(\ell)}_1$ is the same as $\TTP_4^{(\ell+1)}$ (or, in the case that $\ell = t$, the same as $\TTP''$), except for the generation of the state $\Gamma_{\ell}$ (see Algorithm~\ref{alg:ttp-key-gen}, line~\ref{line:gadget-gen}). In $\TTP^{(\ell)}_1$, all classical information encrypted under $pk_{\ell}$ is replaced with encryptions of zeros. In particular, for $i \geq \ell$, line~\ref{line:gadget-gen} is adapted to
\begin{align*}
	\Gamma_i \gets 
	&\MAC.\Sign(\HE.\Enc_{pk_i}(00\cdots 0))\\
	& \otimes \TTP''.\Enc'(sk', \gmid_i) \otimes \TTP.\Enc(sk, \gin_i \otimes \gout_i)
\end{align*}
	where $\TTP''.\Enc'$ also appends a signed encryption of zeros, effectively replacing line~\ref{line:ttp-enc} in Algorithm~\ref{alg:ttp-enc} with
	\[\enc{\sigma} \gets \displaystyle \sum_{x,z \in \{0,1\}^{3m}} \Big(\TC.\Enc((\pi,x,z),\sigma) \otimes \MAC.\Sign_k(\HE.\Enc_{pk}(00\cdots0))\Big)
	\]
	It is important to note that in both $\KeyGen$ and $\Enc'$, the information that is sent to $\VerDec$ through the classical side channel is \emph{not} replaced with zeros. Hence, the structural and encryption information about $\Gamma_{\ell}$ is kept from the adversary, and instead is directly sent (only) to the verification procedure. Whenever $\VerDec$ needs this information, it is taken directly from this trusted source, and the all-zero string sent by the adversary will be ignored.
	\arxiv{
	\begin{figure}
	\centering
	\begin{tikzpicture}
\foreach \x in {0,...,5}
	\filldraw[fill=black] (\x,0) circle (2pt);
\draw[decorate,decoration={snake,amplitude=1pt,segment length=5pt}] (0,0) to[bend right] (2,0);
\draw[decorate,decoration={snake,amplitude=1pt,segment length=5pt}] (1,0) to[bend right] (5,0);
\draw[decorate,decoration={snake,amplitude=1pt,segment length=5pt}] (3,0) to[bend right] (4,0);
\filldraw[fill=white] (3.3,-0.3) rectangle (3.7,0.1);
\node at (3.5,-0.1) {$\P$};
\node[anchor=east] at (-0.5,0) {$\HE.\Enc_{pk_{\ell+1}}(00\cdots0)$ + };
\node at (0,0.75) {$\gin_{\ell}$};
\node at (2.5,0.75) {$\gmid_{\ell}$};
\node at (5,0.75) {$\gout_{\ell}$};
\draw[decoration={brace,raise=8pt,amplitude=5pt}, decorate] (0.75,0) -- (4.25,0);
\end{tikzpicture}
	\caption{In $\TTP_1^{\ell}$, all classically encrypted information for the $\ell$th gadget is replaced by zeros. The quantum state remains the same as in $\TTP$ (see Figure~\ref{fig:ttp-gadget}).}\label{fig:ttp-gadget-1}
	\end{figure}
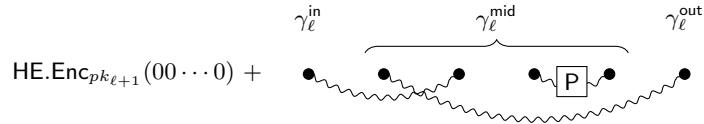
	}{}
	
\arxiv{\item}{\textbf{2.}} $\TTP_2^{(\ell)}$ is the same as $\TTP^{(\ell)}_1$, except that for the $\ell$th gadget, the procedure $\TTP.\Post\GadgetGen$ is called instead of $\TTP.\GadgetGen$:

\begin{algorithm}
	\caption{$\TTP.\Post\GadgetGen(sk_i)$}\label{alg:ttp-post-gadget-gen}
	\begin{algorithmic}[1]
		\State $g_i \gets 0^{|g(sk_i)|}$
		\State $(\gin, \gmid, \gout) \gets$ halves of EPR pairs (send other halves to VerDec)
		\State \textbf{return} $(g_i, \gin_i, \gmid_i, \gout_i)$
	\end{algorithmic}
\end{algorithm}

\todo[inline]{In line 1, we can also choose to give $g(sk_i)$, because it is replaced by zeroes anyway in keygen -- see previous hybrid.}

This algorithm produces a `gadget' in which all qubits are replaced with halves of EPR pairs. These still get encrypted in line~\ref{line:gadget-gen} of Algorithm~\ref{alg:ttp-key-gen}. All other halves of these EPR pairs are sent to $\VerDec$ through the provided quantum channel. $\TTP_2^{(\ell)}.\VerDec$ has access to the structural information $g_{\ell}$ (as this is sent via the classical side information channel from $\KeyGen$ to $\VerDec$) and performs the necessary Bell measurements to recreate $\gin_{\ell}$, $\gmid_{\ell}$ and $\gout_{\ell}$ after the adversary has interacted with the EPR pair halves. Effectively, this postpones the generation of the gadget structure to decryption time. Of course, the measurement outcomes are taken into account by $\VerDec$ when calculating updates to the quantum one-time pad. As can be seen from the description of $\TTP_4^{(\ell)}$, all corrections that follow the $\ell$th one are unaffected by the fact that the server cannot hold the correct information about these postponed measurements, not even in encrypted form.

\arxiv{
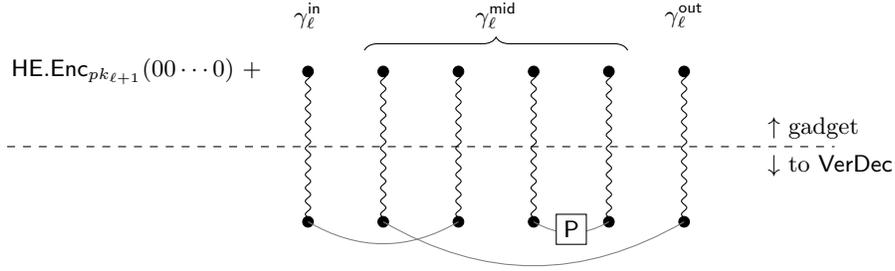
\begin{figure}
\centering
\begin{tikzpicture}
\foreach \x in {0,...,5}
{
	\filldraw[fill=black] (\x,0) circle (2pt);
	\filldraw[fill=black] (\x,-2) circle (2pt);
	\draw[decorate,decoration={snake,amplitude=1pt,segment length=5pt}] (\x,0) to (\x,-2);
}
\draw[dashed] (-4,-1) -- (7,-1);
\node[anchor=south west] at (6,-1) {$\uparrow$ gadget};
\node[anchor=north west] at (6,-1) {$\downarrow$ to $\VerDec$};
\draw[gray] (0,-2) to[bend right] (2,-2);
\draw[gray] (1,-2) to[bend right] (5,-2);
\draw[gray] (3,-2) to[bend right] (4,-2);
\filldraw[fill=white] (3.3,-2.3) rectangle (3.7,-1.9);
\node at (3.5,-2.1) {$\P$};
\node[anchor=east] at (-0.5,0) {$\HE.\Enc_{pk_{\ell+1}}(00\cdots0)$ + };
\node at (0,0.75) {$\gin_{\ell}$};
\node at (2.5,0.75) {$\gmid_{\ell}$};
\node at (5,0.75) {$\gout_{\ell}$};
\draw[decoration={brace,raise=8pt,amplitude=5pt}, decorate] (0.75,0) -- (4.25,0);
\end{tikzpicture}
\caption{In $\TTP_2^{(\ell)}$, the quantum state that consitutes the $\ell$th gadget is replaced with halves of EPR pairs. The other halves are sent to $\VerDec$, where Bell measurements (the gray lines) and the phase gate $\P$ are applied after evaluation.}\label{fig:ttp-gadget-2}
\end{figure}
}{}
	
\arxiv{\item}{\textbf{3.}}$\TTP_3^{(\ell)}$ is the same as $\TTP^{(\ell)}_2$, except that gadget generation for the $\ell$th gadget is handled by $\TTP.\Fake\GadgetGen$ instead of $\TTP.\Post\GadgetGen$.

\begin{algorithm}
	\caption{$\TTP.\Fake\GadgetGen(sk_i)$}\label{alg:ttp-fake-gadget-gen}
	\begin{algorithmic}[1]
		\State $g_i \gets 0^{|g(sk_i)|}$
		\State $(\gin, \gmid, \gout) \gets$ halves of EPR pairs (send other halves to VerDec)
		\State Send $\gmid$ to VerDec as well
		\State \textbf{return} $(g_i, \gin_i, \ket{00\cdots0}, \gout_i)$
	\end{algorithmic}
\end{algorithm}

This algorithm prepares, instead of halves of EPR pairs, $\ket0$-states of the appropriate dimension for $\gmid_{\ell}$. (Note that this dimension does not depend on $sk_{\ell-1}$). For $\gin_{\ell}$ and $\gout_{\ell}$, halves of EPR pairs are still generated, as in $\TTP^{(\ell)}_2$.  Via the side channel, the full EPR pairs for $\gmid_{\ell}$ are sent to $\VerDec$. As in the previous hybrids, the returned gadget is encrypted in $\TTP.\KeyGen$.
	
	$\TTP_3^{(\ell)}.\VerDec$ verifies that the adversary performed the correct Bell measurements on the fake $\ell$th gadget by calling $\TC.\VerDec$. If this procedure accepts, $\TTP_3^{(\ell)}.\VerDec$ performs the verified Bell measurements on the halves of the EPR pairs received from $\TTP_3^{(\ell)}.\KeyGen$ (and subsequently performs the Bell measurements that depend on $g_{\ell}$ on the other halves, as in $\TTP_2^{(\ell)}$). Effectively, $\TTP_3^{(\ell)}.\VerDec$ thereby performs a protocol for $\HE.\Dec$, removing the phase error in the process.
	
\arxiv{
\begin{figure}[h]
\centering
\begin{tikzpicture}
\foreach \x in {0,5}
{
	\filldraw[fill=black] (\x,0) circle (2pt);
	\filldraw[fill=black] (\x,-2) circle (2pt);
	\draw[decorate,decoration={snake,amplitude=1pt,segment length=5pt}] (\x,0) to (\x,-2);
}
\foreach \x in {1,...,4}
{
	\node at (\x,0) {$\ket0$};
	\filldraw[fill=black] (\x,-2) circle (2pt);
	\filldraw[fill=black] (\x,-3) circle (2pt);
	\draw[decorate,decoration={snake,amplitude=1pt,segment length=5pt}] (\x,-2) to (\x,-3);
}
\draw[dashed] (-4,-1) -- (7,-1);
\node[anchor=south west] at (6,-1) {$\uparrow$ gadget};
\node[anchor=north west] at (6,-1) {$\downarrow$ to $\VerDec$};
\draw[gray] (0,-2) to[out=-80,in=-135,looseness=1.5] (2,-3);
\draw[gray] (1,-3) to[out=-25, in=-105,looseness=1.3] (5,-2);
\draw[gray] (3,-3) to[bend right] (4,-3);
\filldraw[fill=white] (3.3,-3.3) rectangle (3.7,-2.9);
\node at (3.5,-3.1) {$\P$};
\node[anchor=east] at (-0.5,0) {$\HE.\Enc_{pk_{\ell+1}}(00\cdots0)$ + };
\node at (0,0.75) {$\gin_{\ell}$};
\node at (2.5,0.75) {$\gmid_{\ell}$};
\node at (5,0.75) {$\gout_{\ell}$};
\draw[decoration={brace,raise=8pt,amplitude=5pt}, decorate] (0.75,0) -- (4.25,0);
\end{tikzpicture}
\caption{In $\TTP_3^{(\ell)}$, all of $\gmid_{\ell}$ is replaced with dummy qubits. $\VerDec$ verifies the Bell measurements performed on these dummy qubits, and performs them on the top halves of the corresponding EPR pairs. Like in $\TTP_2^{(\ell)}$, $\VerDec$ also performs Bell measurements and a $\P$ gate on the lower halves.}\label{fig:ttp-gadget-3}
\end{figure}
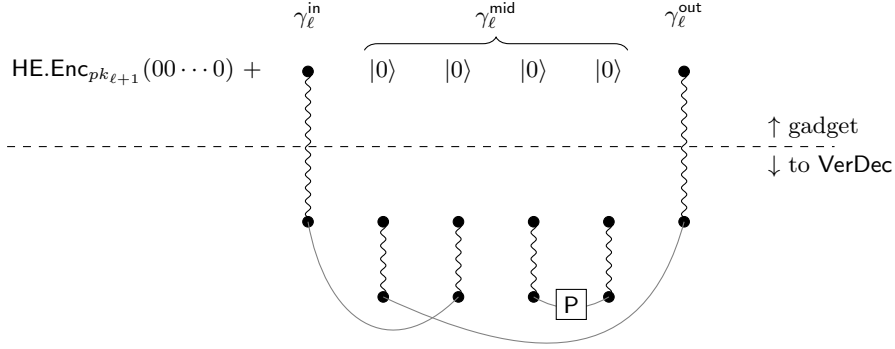
}{}
	
\arxiv{\item}{\textbf{4.}} $\TTP_4^{(\ell)}$ is the same as $\TTP^{(\ell)}_3$, except that $\VerDec$ (instead of performing the Bell measurements of the gadget protocol) uses its knowledge of the initial QOTP keys and all intermediate measurement outcomes to compute whether or not a phase correction is necessary after the $\ell$th $\T$ gate. $\TTP_4^{(\ell)}.\VerDec$ then performs this phase correction on the EPR half entangled with $\gin_{\ell}$, followed by a Bell measurement with the EPR half entangled with $\gout_{\ell}$.

\arxiv{\end{enumerate}}{}

\arxiv{
\begin{figure}[h]
\centering
\begin{tikzpicture}
\foreach \x in {0,5}
{
	\filldraw[fill=black] (\x,0) circle (2pt);
	\filldraw[fill=black] (\x,-2) circle (2pt);
	\draw[decorate,decoration={snake,amplitude=1pt,segment length=5pt}] (\x,0) to (\x,-2);
}
\foreach \x in {1,...,4}
{
	\node at (\x,0) {$\ket0$};
}
\draw[dashed] (-4,-1) -- (7,-1);
\node[anchor=south west] at (6,-1) {gadget};
\node[anchor=north west] at (6,-1) {to $\VerDec$};
\draw[gray] (0,-2) to[bend right] (5,-2);
\filldraw[fill=white] (2.2,-2.9) rectangle (2.8,-2.5);
\node at (2.5,-2.7) {$\P?$};
\node[anchor=east] at (-0.5,0) {$\HE.\Enc_{pk_{\ell+1}}(00\cdots0)$ + };
\node at (0,0.75) {$\gin_{\ell}$};
\node at (2.5,0.75) {$\gmid_{\ell}$};
\node at (5,0.75) {$\gout_{\ell}$};
\draw[decoration={brace,raise=8pt,amplitude=5pt}, decorate] (0.75,0) -- (4.25,0);
\end{tikzpicture}
\caption{In $\TTP_4^{(\ell)}$, the state that the evaluator receives is exactly equal to the state in $\TTP_3^{(\ell)}$ (see Figure~\ref{fig:ttp-gadget-3}). The only difference is the way $\VerDec$ applies the $\P$ gate (conditionally): instead of emulating the gadget usage, $\VerDec$ directly computes whether or not a phase needs to be applied, and performs the teleportation measurement on $\gin_{\ell}$ and $\gout_{\ell}$ accordingly.}\label{fig:ttp-gadget-4}
\end{figure}
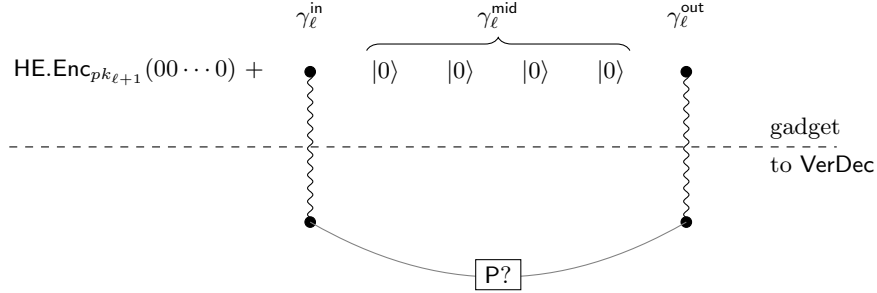
}{}

The first $\ell-1$ gadgets in $\TTP_1^{(\ell)}$ through $\TTP_4^{(\ell)}$ are always functional gadgets, as in $\TTP$. The last $t - \ell$ gadgets are all completely replaced by dummy states, and their functionality is completely outsourced to $\VerDec$. In four steps described above, the functionality of the $\ell$th gadget is also transferred to $\VerDec$. It is important to replace only one gadget at a time, because replacing a real gadget with a fake one breaks the functionality of the gadgets that occur later in the evaluation: the encrypted classical information held by the server does not correspond to the question of whether or not a phase correction is needed. By completely outsourcing the phase correction to $\VerDec$, as is done for all gadgets after the $\ell$th one in all $\TTP_i^{(\ell)}$ schemes, we ensure that this incorrect classical information does not influence the outcome of the computation. Hence, correctness is maintained throughout the hybrid transformations. We now show that these transformations of the scheme do not significantly affect the adversary's winning probability in the hybrid indistinguishability game.

\begin{lemma}\label{lem:hyb-remove-pk}
	For any QPT $\advA$, there exists a negligible function $\negl$ such that for all $1 \leq \ell \leq t$,\\$\AdvHyb{\TTP_1^{(\ell)}}{\TTP_4^{(\ell+1)}} \leq \negl(\kappa)$.
\end{lemma}

\begin{proof}[sketch]
	In $\TTP_4^{(\ell+1)}$, no information about $sk_{(\ell)}$ is sent to the adversary. In the original $\TTP$ scheme, the structure of the quantum state $\Gamma_{\ell+1}$ depended on it, but this structure has been replaced with dummy states in several steps in $\TTP_2^{\ell+1}$ through $\TTP_4^{\ell+1}$.
	
	This is fortunate, since if absolutely no secret-key information is present, we are able to bound the difference in winning probability between $\HybridGame{\TTP_4^{(\ell+1)}}$ and $\HybridGame{\TTP_1^{\ell}}$ by reducing it to the IND-CPA security against quantum adversaries~\cite{BJ15} of the classical homomorphic encryption scheme $\HE$.
	
	The proof is closely analogous to the proof of Lemma 1 in~\cite{DSS16}, and on a high level it works as follows. Let $\advA = (\advA_1, \advA_2, \advA_3)$ be a QPT adversary for the game $\HybridGame{\TTP_1^{(\ell)}}$ or $\HybridGame{\TTP_4^{(\ell+1)}}$ (we do not need to specify for which one, since they both require the same input/output interface). A new quantum adversary $\advA'$ for the classical IND-CPA indistinguishability game is defined by having the adversary taking the role of challenger in either the game $\HybridGame{\TTP_1^{(\ell)}}$ or the game $\HybridGame{\TTP_4^{(\ell+1)}}$. Which game is simulated depends on the coin flip of the challenger for the IND-CPA indistinguishability game, and is unknown to $\advA'$. This situation is achieved by having $\advA'$ send any classical plaintext that should be encrypted under $pk_{\ell}$ to the challenger, so that either that plaintext is encrypted or a string of zeros is.
	
	Based on the guess of the simulated $\advA$, which $\advA'$ can verify to be correct or incorrect in his role of challenger, $\advA'$ will guess which of the two games was just simulated. By IND-CPA security of the classical scheme against quantum adversaries, $\advA'$ cannot succeed in this guessing game with nonnegligible advantage over random guessing. This means that the winning probability of $\advA$ in both games cannot differ by a lot. For details, we refer the reader the proof of Lemma~\ref{lem:hyb-remove-gadget}, in which a very similar approach is taken.
	
	Technically, the success probability of $\advA'$, and thus the function $\negl$, may depend on $\ell$. A standard randomizing argument, as found in e.g. the discussion of hybrid arguments in~\cite{KL14}, allows us to get rid of this dependence by defining another adversary $\advA''$ that selects a random value of $j$, and then bounding the advantage of $\advA''$ by a negligible function that is independent of $j$.
	\todo[inline]{(Possibly) extend sketch to full proof?}
	\qed
\end{proof}

\begin{lemma}\label{lem:hyb-postpone-gadget}
	For $1 \leq \ell \leq t$ and any QPT $\advA$, $\AdvHyb{\TTP_1^{(\ell)}}{\TTP_2^{(\ell)}} = 0$.
\end{lemma}

\begin{proof}
	In $\TTP_1^{(\ell)}$, the $\ell$th error-correction gadget consists of a number of EPR pairs arranged in a certain order, as described by the garden-hose protocol for $\HE.\Dec$. For example, this protocol may dictate that the $i$th and $j$th qubit of the gadget must form an EPR pair $\ket{\Phi^+}$ together. This can alternatively be achieved by creating two EPR pairs, placing half of each pair in the $i$th and $j$th position of the gadget state, and performing a Bell measurement on the other two halves. This creates a Bell pair $\X^a\Z^b\ket{\Phi^+}$ in positions $i$ and $j$, where $a,b \in \{0,1\}$ describe the outcome of the Bell measurement.
	
	From the point of view of the adversary, it does not matter whether these Bell measurements are performed during $\KeyGen$, or whether the halves of EPR pairs are sent to $\VerDec$ for measurement -- because the key to the quantum one-time pad of the $\ell$th gadget is not sent to the adversary at all, the same state is created with a completely random Pauli in either case. Of course, the teleportation correction Paulis of the form $\X^a\Z^b$ need to be taken into account when updating the keys to the quantum one-time pad on the data qubits after the gadget is used. $\VerDec$ has all the necessary information to do this, because it observes the measurement outcomes, and computes the key updates itself (instead of decrypting the final keys from the computation log).
	
	Thus, with the extra key update steps in $\TTP_2^{(\ell)}.\VerDec$, the inputs to the adversary are exactly the same in the games of $\TTP_1^{(\ell)}$ and $\TTP_2^{(\ell)}$.
	\qed
\end{proof}

\begin{lemma}\label{lem:hyb-remove-gadget}
	For any QPT $\advA$, there exists a negligible function $\negl$ such that for all $1 \leq \ell \leq t$,\\ $\AdvHyb{\TTP_2^{(\ell)}}{\TTP_3^{(\ell)}} \leq \negl(\kappa)$.
\end{lemma}

\begin{proof}
	We show this by reducing the difference in winning probabilities in the statement of the lemma to the IND-VER security of the somewhat homomorphic scheme $\TC$. Intuitively, because $\TC$ is IND-VER, if $\TTP_2^{(\ell)}$ accepts the adversary's claimed circuit of Bell measurements on the EPR pair halves, the effective map on those EPR pairs is the claimed circuit. Therefore, we might just as well ask $\VerDec$ to apply this map, as we do in $\TTP_3^{(\ell)}$, to get the same output state. If $\TTP_2^{(\ell)}$ rejects the adversary's claimed circuit on those EPR pair halves, then $\TTP_3^{(\ell)}$ should reject too. This is why we let the adversary act on an encrypted dummy state of $\ket0$s.
	
	Let $\advA = (\advA_1, \advA_2,\advA_3)$ be a set of QPT algorithms on the appropriate registers, so that we can consider it as an adversary for the hybrid indistinguishability game for either $\TTP_2^{(\ell)}$ or $\TTP_3^{(\ell)}$ (see Definition~\ref{def:hybrid-game}). Note the input/output wires to the adversary in both these games are identical, so we can evaluate $\Pr[\HybridGame{\TTP_2^{(\ell)}} = 1]$ and $\Pr[\HybridGame{\TTP_3^{(\ell)}} = 1]$ for the same $\advA$.
	
	Now define an adversary $\advA' = (\advA'_1, \advA'_2, \advA'_3)$ for the IND-VER game against $\TC$, $\VerGamePrime{\TC}$, as follows:
	
\textbf{1.} $\advA'_1$:
Run $\TTP_2^{(\ell)}.\KeyGen$ until the start of line~\ref{line:gadget-gen} in the $\ell$th iteration of that loop. Up to this point, $\TTP_2^{(\ell)}.\KeyGen$ is identical to $\TTP_3^{(\ell)}.\KeyGen$. It has generated real gadgets $\Gamma_1$ through $\Gamma_{\ell-1}$, and halves of EPR pairs for $\gin_{\ell}$, $\gmid_{\ell}$ and $\gout_{\ell}$. Note furthermore that the permutation $\pi_{\ell}$ is used nowhere. Now send $\gmid_{\ell}$ to the challenger via the register $X$, and everything else (including $sk$) to $\advA'_2$ via the side register $R$.

\textbf{2.} $\advA'_2$:
		Continue $\TTP_2^{(\ell)}.\KeyGen$ using the response from the challenger instead of $\TTP.\Enc'(sk',\gmid_{\ell})$ on line~\ref{line:gadget-gen} in the $\ell$th iteration. Call the result $\rho_{\evk}$. Again, this part of the key generation procedure is identical for $\TTP_2^{(\ell)}$ and $\TTP_3^{(\ell)}$.
Start playing the hybrid indistinguishability game with $\advA$:
			\begin{itemize}
				\item Flip a bit $r \in \{0,1\}$.
				\item Send $\rho_{evk}$ to $\advA_1$. If $r = 0$, encrypt the response of $\advA_1$ using the secret key $sk$ generated by $\advA'_1$. Note that for this, the permutation $\pi_{\ell}$ is also not needed. If $r = 1$, encrypt a $\ket0$ state of appropriate dimension instead.
				\item Send the resulting encryption, along with the side info from $\advA_1$, to $\advA_2$.
				\item On the output of $\advA_2$, start running $\TTP_2^{(\ell)}.\VerDec$ until the actions on the $\ell$th gadget need to be verified. Since the permutation on the state $\gmid_{\ell}$ is unknown to $\advA'_2$ (it was sent to the challenger for encryption), it cannot verify this part of the computation.
				\item Instead, send the relevant part of the computation log to the challenger for verification, along with the relevant part of the claimed circuit (the Bell measurements on the gadget state), and the relevant qubits, all received from $\advA_2$, to the challenger for verification and decryption.
				\item In the meantime, send the rest of the working memory to $\advA'_3$ via register $R'$.
			\end{itemize}
			
\textbf{3.} $\advA'_3$:
Continue the simulation of the hybrid game with $\advA$:
			\begin{itemize}
				\item If the challenger rejects, reject and replace the entire quantum state by the fixed dummy state $\Omega$.
				\item If the challenger accepts, then we know that the challenger applies the claimed subcircuit to the quantum state it did not encrypt (either $\ket0$ or $\gmid_{\ell}$), depending on the bit the challenger flipped), and possibly swaps this state back in (again depending on which bit it flipped). Continue the $\TTP_2^{(\ell)}.\VerDec$ computation for the rest of the computation log.
				\item Send the result (the output quantum state, the claimed circuit, and the accept/reject flag) to $\advA_3$, and call its output bit $r'$.
			\end{itemize}
Output 0 if $r = r'$, and 1 otherwise. (i.e., output $NEQ(r,r')$)
\todo[inline]{Add image still?}

Recall from Definition~\ref{def:hybrid-game} that the challenger flips a coin (let us call the outcome $s \in \{0,1\}$) to decide whether to encrypt the quantum state provided by $\advA'$, or to swap in an all-zero dummy state before encrypting. Keeping this in mind while inspecting the definition of $\advA'$, one can see that whenever $s = 0$, $\advA'$ takes the role of challenger in the game $\HybridGame{\TTP_2^{(\ell)}}$ with $\advA$, and whenever $s = 1$, they play $\HybridGame{\TTP_3^{(\ell)}}$. Now let us consider when the newly defined adversary $\advA'$ wins the VER indistinguishability game for $\TC$. If $s = 0$, $\advA'$ needs to output a bit $s' = 0$ to win. This happens, by definition of $\advA'$, if and only if $\advA$ wins the game $\HybridGame{\TTP_2^{(\ell)}}$ (i.e.\ $r = r'$). On the other hand, if $s = 1$, $\advA'$ needs to output a bit $s' = 1$ to win. This happens, by definition of $\advA'$, if and only if $\advA$ loses the game $\HybridGame{\TTP_3^{(\ell)}}$ (i.e.\ $r \neq r'$). Thus the winning probability of $\advA'$ is:
\begin{align*}
&\Pr[\VerGamePrime{\TC} = 1] =\\
&= \Pr[s = 0]\cdot\Pr[\HybridGame{\TTP_2^{(\ell)}} = 1] + \Pr[s = 1]\cdot\Pr[\HybridGame{\TTP_3^{(\ell)}} = 0]\\
&= \frac{1}{2}\Pr[\HybridGame{\TTP_2^{(\ell)}} = 1] + \frac{1}{2}\Big(1  - \Pr[\HybridGame{\TTP_3^{(\ell)}} = 1]\Big)\\
&= \frac{1}{2} + \frac{1}{2}\Big(\Pr[\HybridGame{\TTP_2^{(\ell)}} = 1] - \Pr[\HybridGame{\TTP_3^{(\ell)}} = 1]\Big)
\end{align*}
From the IND-VER property of $\TC$ (see Theorem~\ref{thm:tc-security-2}) we know that the above is at most $\frac{1}{2} + \negl(\kappa)$. From this (and a randomizing argument similar to Lemma~\ref{lem:hyb-remove-pk}), the statement of the lemma follows directly.
\qed
\end{proof}

\begin{lemma}\label{lem:hyb-bypass-gadget}
	For any QPT $\advA$, there exists a negligible function $\negl$ such that for all $1 \leq \ell \leq t$,\\ $\AdvHyb{\TTP_3^{(\ell)}}{\TTP_4^{(\ell)}} \leq \negl(\kappa)$.
\end{lemma}

\begin{proof} Let $f(s)$ be the bit that, after the $\ell$th $\T$ gate, determines whether or not a phase correction is necessary. Here, $s$ is all the relevant starting information (such as quantum one-time pad keys, gadget structure, permutations, and applied circuit), and $f$ is some function that determines the $\X$ key on the relevant qubit right before application of the $\T$ gate.
	
	In $\TTP_3^{(\ell)}$, a phase correction after the $\ell$th $\T$ gate is applied conditioned on the outcome of
	\[
	\HE.\Dec_{sk_{\ell-1}}(\HE.\Eval_{evk_{0}, ..., evk_{\ell-1}}^f(\HE.\Enc_{pk_0}(s))),
	\]
	because the garden-hose computation in the gadget computes the classical decryption.
	In the above expression, we again slightly abuse notation, as in the proof of Lemma~\ref{lem:remove-log}, and include recryption steps in $\HE.\Eval_{evk_0, ..., \evk_{\ell-1}}$. As long as $t$ is polynomial in $\kappa$, we have, by correctness of $\HE$,
	\[
	\Pr[\HE.\Dec_{sk_{\ell-1}}(\HE.\Eval_{evk_{0}, ..., evk_{\ell-1}}^f(\HE.\Enc_{pk_0}(s))) \neq f(s)] \leq \negl(\kappa).
	\]
	In $\TTP_4^{(\ell)}$, the only difference from $\TTP_3^{(\ell)}$ is that, instead of performing the garden-hose computation on the result of the classical homomorphic evaluation procedure, the phase correction is applied directly by $\VerDec$, conditioned on $f(s)$. The probability that in $\TTP_4^{(\ell)}$, a phase is applied (or not) when in $\TTP_3^{(\ell)}$ it is not (or is), is negligible. The claim follows directly.
\qed
\end{proof}

\subsubsection{Final Hybrid: Removing All Classical FHE.}

In $\TTP_4^{(1)}$, all of the error-correction gadgets have been removed from the evaluation key, and the error-correction functionality has been redirected to $\VerDec$ completely. Effectively, $\TTP_4^{(1)}.\KeyGen$ samples a permutation $\pi$, generates a lot of magic states (for $\P$, $\H$ and $\T$) and encrypts them using $\TC.\Enc_{\pi}$, after which the keys to the quantum one-time pad used in that encryption are homomorphically encrypted under $pk_0$. The adversary is allowed to act on those encryptions, but while its homomorphic computations are syntactically checked in the log, $\VerDec$ does not decrypt and use the resulting values. This allows us to link $\TTP_4^{(1)}$ to a final hybrid, $\TTP^f$, where all classical information is replaced with zeros before encrypting.

The proof of the following lemma is analogous to that of Lemma~\ref{lem:hyb-remove-pk}, and reduces to the IND-CPA security of the classical scheme $\HE$:

\begin{lemma}\label{lem:remove-final-pk}
	For any QPT $\advA$, $\AdvHyb{\TTP_4^{(1)}}{\TTP^f} \leq \negl(\kappa)$.
\end{lemma}

\subsubsection{Proof of main theorem.}

Considering $\TTP^f$ in more detail, we can see that it is actually very similar to $\TC$. This allows us to prove the following lemma, which is the last ingredient for the proof of verifiability of $\TTP$.

\begin{lemma}\label{lem:ttp-to-tc}
	For any QPT $\advA$,
$\Pr[\HybridGame{\TTP^f} = 1] \leq \frac{1}{2} + \negl(\kappa).$
\end{lemma}

\begin{proof}
To see the similarity with $\TC$, consider the four algorithms of $\TTP^f$.

In $\TTP^f.\KeyGen$, a permutation $\pi$ is sampled, and magic states for $\P$, $\H$ and $\T$ are generated, along with some EPR pair halves (to replace $in_i$ and $out_i$). For all generated quantum states, random keys for QOTPs are sampled, and the states are encrypted using $\TC.\Enc$ with these keys as secret keys. No classical FHE is present anymore. Thus, $\TTP^f.\KeyGen$ can be viewed as $\TC.\KeyGen$, followed by $\TC.\Enc$ on the magic states and EPR pair halves.

$\TTP^f.\Enc$ is identical to $\TC.\Enc$, only the keys to the quantum one-time pad are sampled on the fly and sent to $\TTP^f.\VerDec$ via a classical side-channel, whereas $\TC.\VerDec$ receives them as part of the secret key. Since the keys are used exactly once and not used anywhere else besides in $\Enc$ and $\VerDec$, this difference does not affect the outcome of the game.

$\TTP^f.\Eval$ only requires $\CNOT$, classically controlled Paulis, computational basis measurements and Hadamard basis measurements. For the execution of any other gate, it suffices to apply a circuit of $\CNOT$, classically controlled Paulis, and measurements to the encrypted data, encrypted magic states and/or encrypted EPR halves.

$\TTP^f.\VerDec$ does two things: (i) it syntactically checks the provided computation log, and (ii) it runs $\TC.\VerDec$ to verify that the evaluation procedure correctly applied the circuit of $\CNOT$s and measurements.

An execution of $\HybridGame{\TTP^f}$ for any $\advA$ corresponds to the two-round VER indistinguishability game for $\TC$ as follows. Let $\advA = (\advA_1, \advA_2, \advA_3)$ be a polynomial-time adversary for the game $\HybridGame{\TTP^f}$. Define an additional QPT $\advA_0$ that produces magic states and EPR pair halves to the register $X_1$. The other halves of the EPR pairs are sent through $R$, and untouches by $\advA_1$ and $\advA_2$. The above analysis shows that the adversary $\advA' = (\advA_0, \advA_1, \advA_2, \advA_3)$ can be viewed as an adversary for the VER-2 indistinguishability game $\VerTwoGamePrime{\TC}$ and wins whenever $\HybridGame{\TTP^f} = 1$. The other direction does not hold: $\advA$ loses the hybrid indistinguishability game if $\TTP^f.\VerDec$ rejects check (i), but accepts check (ii) (see above). In this case, $\advA'$ would still win the VER-2 indistinguishability game. Hence,
\[
\Pr[\HybridGame{\TTP^f} = 1] \leq \Pr[\VerTwoGamePrime{\TC} = 1].
\]
Theorem~\ref{thm:tc-security-2} yields $\Pr[\VerTwoGamePrime{\TC} = 1] \leq \frac{1}{2} + \negl(\kappa)$, and the result follows.
\qed
\end{proof}

Now we finally have all the ingredients needed to prove our main theorem:

\begin{theorem}\label{thm:verif-ttp}
	The vQFHE scheme $\TTP$ satisfies $\kappa$-SEM-VER.
\end{theorem}

\begin{proof}
	From Lemmas~\ref{lem:remove-MAC}, \ref{lem:remove-log}, \ref{lem:hyb-remove-pk}, \ref{lem:hyb-postpone-gadget}, \ref{lem:hyb-remove-gadget}, \ref{lem:hyb-bypass-gadget}, and \ref{lem:remove-final-pk}, we may conclude that if $t$ (the number of $\T$ gates in the circuit) is polynomial in $\kappa$ (the security parameter), then for any polynomial-time adversary $\advA$,
	\[
	\Pr[\VerGame{\TTP} = 1] \ \ - \ \ \Pr[\HybridGame{\TTP^f} = 1] \ \ \leq \ \ \negl(\kappa),
	\]
	since the sum poly-many negligible terms is negligible (it is important to note that there is only a constant number of \emph{different} negligible terms involved). By Lemma~\ref{lem:ttp-to-tc}, which reduces verifiability of $\TTP^f$ to verifiability of $\TC$, $\Pr[\HybridGame{\TTP^f} = 1] \leq 1/2 + \negl(\kappa)$. It follows that $\Pr[\VerGame{\TTP} = 1] \leq 1/2 + \negl(\kappa)$, i.e., that $\TTP$ is $\kappa$-IND-VER. By Theorem~\ref{thm:IND-SEM}, $\TTP$ is also $\kappa$-SEM-VER.
\qed
\end{proof}

\section{Application to quantum one-time programs}

\subsubsection{One-time programs.}

We now briefly sketch an application of the vQFHE scheme to one-time programs. A classical one-time program (or cOTP) is an idealized object which can be used to execute a function once, but then self-destructs.  In the case of a quantum OTP (or qOTP), the program executes a quantum channel $\Phi$. In the usual formalization, $\Phi$ has two inputs and is public. One party (the sender) creates the qOTP by fixing one input, and the qOTP is executed by a receiver who selects the other input. To recover the intuitive notion of OTP, choose $\Phi$ to be a universal circuit. We will work in the universally-composable (UC) framework, following the approach of~\cite{BGS13}. We thus first define the ideal functionality of a qOTP.
\begin{definition}[Functionality 3 in~\cite{BGS13}]
The ideal functionality $\mathcal F_\Phi^\OTP$ for a channel $\Phi_{XY \to Z}$ is the following:
\begin{enumerate}
\item \textbf{\emph{Create:}} given register $X$ from sender, store $X$ and send $\create$ to receiver.
\item \textbf{\emph{Execute:}} given register $Y$ from receiver, send $\Phi$ applied to $XY$ to receiver. Delete any trace of this instance.
\end{enumerate}
\end{definition}
A qOTP is then a real functionality which ``UC-emulates'' the ideal functionality~\cite{Unruh10}. As in~\cite{BGS13}, we only allow corrupting receivers; unlike~\cite{BGS13}, we consider computational (rather than statistical) UC security. The achieved result is therefore slightly weaker. The construction within our vQFHE framework is however much simpler, and shows the relative ease with which applications of vQFHE can be constructed.

\subsubsection{The construction.}

Choose a vQFHE scheme $\Pi = (\KeyGen, \Enc, \Eval, \VerDec)$ satisfying SEM-VER. For simplicity, we first describe the classical input/output case, i.e., the circuit begins and ends with full measurement of all qubits. Let $C$ be such a circuit, for the map $\Phi_{XY \to Z}$. On \textbf{Create}, the sender generates keys $(k, \rho_\evk) \from \KeyGen$ and encrypts their input register $X$ using $k$. The sender also generates a classical OTP for the public, classical function $\VerDec$, choosing the circuit and key inputs to be $C$ and $k$; the computation log is left open for the receiver to select. The qOTP is then the triple 
$$
\Xi_C^X := \left(\rho_\evk, \Enc_k(\rho_X), \OTP_{\VerDec}(C, k)\right) .
$$
On \textbf{Execute}, the receiver computes as follows. The receiver's (classical) input $Y$ together with the (public) circuit $C$ defines a homomorphic computation on the ciphertext $\Enc_k(\rho_X)$, which the receiver can perform using $\Eval$ and $\rho_\evk$. Since $C$ has only classical outputs, the receiver measures the final state completely. At the end of that computation, the receiver holds the (completely classical) output of the computation log from $\Eval$. The receiver plugs the log into $\OTP_{\VerDec}(C, k)$, which produces the decrypted output.

We handle the case of arbitrary circuits $C$ (with quantum input and output) as follows. Following the ideas of~\cite{BGS13}, we augment the above quantum OTP with two auxiliary quantum states: an ``encrypt-through-teleport'' gadget $\sigma_\textsf{in}$ and a ``decrypt-through-teleport'' gadget $\sigma_\textsf{out}$. These are maximally entangled states with the appropriate map (encrypt or decrypt) applied to one half. The receiver uses teleportation on $\sigma^\textsf{in}_{Y_1W_1}$ to encrypt their input register $Y$ before evaluating, and places the teleportation measurements into the computation log. After evalution, the receiver uses $\sigma^\textsf{out}_{W_2Y_2}$ to teleport the plaintext out, combining the teleportation measurements with the output of $\OTP_{\VerDec}(C, k)$ to compute the final QOTP decryption keys.

\subsubsection{Security proof sketch.}

Starting with a QPT adversary $\algo A$ which attacks the real functionality, we construct a QPT simulator $\algo S$ which attacks the ideal functionality (with similar success probability). We split $\algo A$ into $\algo A_1$ (receive input, output the OTP query and side information) and $\algo A_2$ (receive result of OTP query and side information, produce final output). The simulator $\algo S$ will generate its own keys, provide fake gadgets that will trick $\algo A$ into teleporting its input to $\algo S$, who will then use that input on the ideal functionality. Details follow.

The simulator first generates $(k, \rho_\evk) \from \KeyGen$ and encrypts the input $X$ via $\Enc_k$. Instead of the encrypt gadget $\sigma^\textsf{in}_{Y_1W_1}$, $\algo S$ provides half of a maximally entangled state in register $Y$ and likewise in register $W$. The other  halves $Y'_1$ and $W'_1$ of these entangled states are kept by $\algo S$. The same is done in place of the decrypt gadget $\sigma^\textsf{out}_{W_2Y_2}$, with $\algo S$ keeping $Y_2'$ and $W_2'$. Then $\algo S$ runs $\algo A_1$ with input $\rho_\evk, \Enc_k(\rho_X)$ and registers $Y$ and $W$. It then executes $\VerDec_k$ on the output (i.e., the query) of $\algo A_1$ to see if $\algo A_1$ correctly followed the $\Eval$ protocol. If it did not, then $\algo S$ aborts; otherwise, $\algo S$ plugs register $Y_1'$ into the ideal functionality, and then teleports the output into register $W_2'$. Before responding to $\algo A_2$, it corrects the one-time pad keys appropriately using its teleportation measurements.

\section{Conclusion}
In this work, we devised a new quantum-cryptographic primitive: quantum fully-homomorphic encryption with verification (vQFHE). Using the trap code for quantum authentication~\cite{BGS13} and the garden-hose gadgets of~\cite{DSS16}, we constructed a vQFHE scheme $\TTP$ which satisfies (i.) correctness, (ii.) compactness, (iii.) security of verification, (iv.) IND-CPA secrecy, and (v.) authentication. We also outlined a first application of vQFHE, to  quantum one-time programs. 

We leave open several interesting directions for future research. Foremost is finding more applications of vQFHE. Another interesting question is whether vQFHE schemes exist where verification can be done publicly (i.e., without the decryption key), as is possible classically. Finally, it is unknown whether vQFHE (or even QFHE) schemes exist with evaluation key that does not scale with the size of the circuit at all.

\section{Acknowledgements}

This work was completed while GA was a member of the QMATH center at the Department of Mathematical Sciences at the University of Copenhagen. GA and FS acknowledge financial support from the European Research Council (ERC Grant Agreement no 337603), the Danish Council for Independent Research (Sapere Aude), Qubiz - Quantum Innovation Center, and VILLUM FONDEN via the QMATH Centre of Excellence (Grant No. 10059). CS is supported by an NWO VIDI grant.

\bibliography{obfuscation}

\begin{thebibliography}{22}
\providecommand{\natexlab}[1]{#1}
\providecommand{\url}[1]{\texttt{#1}}
\expandafter\ifx\csname urlstyle\endcsname\relax
  \providecommand{\doi}[1]{doi: #1}\else
  \providecommand{\doi}{doi: \begingroup \urlstyle{rm}\Url}\fi

\bibitem[Aharonov et~al.(2008)Aharonov, {Ben-Or}, and Eban]{ABE08}
Dorit Aharonov, Michael {Ben-Or}, and Elad Eban.
\newblock Interactive proofs for quantum computations.
\newblock \emph{arXiv preprint arXiv:0810.5375}, 2008.

\bibitem[Alagic et~al.()Alagic, Dulek, Schaffner, and Speelman]{supp}
Gorjan Alagic, Yfke Dulek, Christian Schaffner, and Florian Speelman.
\newblock Supplementary material.

\bibitem[Alagic et~al.(2016)Alagic, Broadbent, Fefferman, Gagliardoni,
  Schaffner, and Jules]{ABFGSS16}
Gorjan Alagic, Anne Broadbent, Bill Fefferman, Tommaso Gagliardoni, Christian
  Schaffner, and Michael~St.\ Jules.
\newblock Computational security for quantum encryption.
\newblock In \emph{9th International Conference on Information Theoretic
  Security (ITICS)}, pages 47--71, 2016.
\newblock \doi{10.1007/978-3-319-49175-2_3}.

\bibitem[Barak and Brakerski(2012)]{blogBB12}
Boaz Barak and Zvika Brakerski.
\newblock Windows on theory: The swiss army knife of cryptography, 2012.
\newblock URL
  \url{https://windowsontheory.org/2012/05/01/the-swiss-army-knife-of-cryptography/}.

\bibitem[Brakerski and Vaikuntanathan(2011)]{BV11}
Zvika Brakerski and Vinod Vaikuntanathan.
\newblock Efficient fully homomorphic encryption from (standard) {LWE}.
\newblock In \emph{52nd Annual Symposium on Foundations of Computer Science
  (FOCS)}, pages 97--106, 2011.
\newblock \doi{10.1109/FOCS.2011.12}.

\bibitem[Broadbent and Jeffery(2015)]{BJ15}
Anne Broadbent and Stacey Jeffery.
\newblock Quantum homomorphic encryption for circuits of low {T-gate}
  complexity.
\newblock In \emph{Advances in Cryptology--CRYPTO 2015}, pages 609--629.
  Springer, 2015.

\bibitem[Broadbent and Wainewright(2016)]{BW16}
Anne Broadbent and Evelyn Wainewright.
\newblock Efficient simulation for quantum message authentication.
\newblock \emph{arXiv preprint arXiv:1607.03075}, 2016.

\bibitem[Broadbent et~al.(2009)Broadbent, Fitzsimons, and Kashefi]{BFK09}
Anne Broadbent, Joseph Fitzsimons, and Elham Kashefi.
\newblock Universal blind quantum computation.
\newblock In \emph{50th Annual Symposium on Foundations of Computer Science
  (FOCS)}, pages 517--526. IEEE, 2009.

\bibitem[Broadbent et~al.(2013)Broadbent, Gutoski, and Stebila]{BGS13}
Anne Broadbent, Gus Gutoski, and Douglas Stebila.
\newblock Quantum one-time programs.
\newblock In \emph{Advances in Cryptology--CRYPTO 2013}, pages 344--360.
  Springer, 2013.

\bibitem[Broadbent et~al.(2016)Broadbent, Ji, Song, and Watrous]{BJSW16}
Anne Broadbent, Zhengfeng Ji, Fang Song, and John Watrous.
\newblock Zero-knowledge proof systems for {QMA}.
\newblock In \emph{57th Annual Symposium on Foundations of Computer Science
  (FOCS)}, pages 31--40, Oct 2016.
\newblock \doi{10.1109/FOCS.2016.13}.

\bibitem[Buhrman et~al.(2013)Buhrman, Fehr, Schaffner, and Speelman]{BFSS13}
Harry Buhrman, Serge Fehr, Christian Schaffner, and Florian Speelman.
\newblock The garden-hose model.
\newblock In \emph{Proceedings of the 4th Conference on Innovations in
  Theoretical Computer Science}, pages 145--158. ACM, 2013.
\newblock ISBN 978-1-4503-1859-4.
\newblock \doi{10.1145/2422436.2422455}.
\newblock URL \url{http://doi.acm.org/10.1145/2422436.2422455}.

\bibitem[Dulek et~al.(2016)Dulek, Schaffner, and Speelman]{DSS16}
Yfke Dulek, Christian Schaffner, and Florian Speelman.
\newblock Quantum homomorphic encryption for polynomial-sized circuits.
\newblock In \emph{Advances in Cryptology--CRYPTO 2016}, pages 3--32. Springer,
  2016.
\newblock \doi{10.1007/978-3-662-53015-3_1}.

\bibitem[Dupuis et~al.(2012)Dupuis, Nielsen, and Salvail]{DNS12}
Fr{\'e}d{\'e}ric Dupuis, Jesper~Buus Nielsen, and Louis Salvail.
\newblock Actively secure two-party evaluation of any quantum operation.
\newblock In \emph{Advances in Cryptology--CRYPTO 2012}, pages 794--811.
  Springer, 2012.

\bibitem[Garg et~al.(2013)Garg, Gentry, Halevi, Raykova, Sahai, and
  Waters]{GGHRSW13}
S.~Garg, C.~Gentry, S.~Halevi, M.~Raykova, A.~Sahai, and B.~Waters.
\newblock Candidate indistinguishability obfuscation and functional encryption
  for all circuits.
\newblock In \emph{54th Annual Symposium on Foundations of Computer Science
  (FOCS)}, pages 40--49, Oct 2013.
\newblock \doi{10.1109/FOCS.2013.13}.

\bibitem[Gentry(2009)]{Gentry09}
Craig Gentry.
\newblock Fully homomorphic encryption using ideal lattices.
\newblock In \emph{41st Annual ACM Symposium on Theory of Computing (STOC)},
  pages 169--178, 2009.
\newblock \doi{10.1145/1536414.1536440}.

\bibitem[Katz and Lindell(2014)]{KL14}
Jonathan Katz and Yehuda Lindell.
\newblock \emph{Introduction to modern cryptography}.
\newblock CRC press, 2014.

\bibitem[{Newman} and {Shi}(2017)]{NS17}
M.~{Newman} and Y.~{Shi}.
\newblock {Limitations on Transversal Computation through Quantum Homomorphic
  Encryption}.
\newblock \emph{ArXiv e-prints}, April 2017.

\bibitem[Ouyang et~al.(2015)Ouyang, Tan, and Fitzsimons]{FTO15}
Yingkai Ouyang, {Si-Hui} Tan, and Joseph Fitzsimons.
\newblock Quantum homomorphic encryption from quantum codes.
\newblock \emph{arXiv preprint arXiv:1508.00938}, 2015.

\bibitem[Shor and Preskill(2000)]{SP00}
Peter~W. Shor and John Preskill.
\newblock Simple proof of security of the {BB84} quantum key distribution
  protocol.
\newblock \emph{Phys. Rev. Lett.}, 85:\penalty0 441--444, Jul 2000.
\newblock \doi{10.1103/PhysRevLett.85.441}.

\bibitem[Tan et~al.(2016)Tan, Kettlewell, Ouyang, Chen, and
  Fitzsimons]{CFKOT16}
Si{-}Hui Tan, Joshua~A. Kettlewell, Yingkai Ouyang, Lin Chen, and Joseph
  Fitzsimons.
\newblock {A quantum approach to homomorphic encryption}.
\newblock \emph{Scientific Reports}, 6:\penalty0 33467, September 2016.
\newblock \doi{10.1038/srep33467}.

\bibitem[Unruh(2010)]{Unruh10}
Dominique Unruh.
\newblock Universally composable quantum multi-party computation.
\newblock In \emph{Advances in Cryptology--EUROCRYPT 2010}, pages 486--505,
  2010.
\newblock arXiv:0910.2912 [quant-ph].

\bibitem[Yu et~al.(2014)Yu, P\'erez-Delgado, and Fitzsimons]{FPY14}
Li~Yu, Carlos~A. P\'erez-Delgado, and Joseph~F. Fitzsimons.
\newblock Limitations on information-theoretically-secure quantum homomorphic
  encryption.
\newblock \emph{Phys. Rev. A}, 90:\penalty0 050303, 2014.
\newblock \doi{10.1103/PhysRevA.90.050303}.

\end{thebibliography}

\appendix

\section{Equivalence of $\kappa$-IND-VER and $\kappa$-SEM-VER.}\label{app:IND-SEM}

In Lemma~\ref{thm:IND-SEM}, it was shown that if a scheme is $\kappa$-IND-VER, then it is also $\kappa$-SEM-VER. We here provide the full proof of Lemma~\ref{thm:IND-SEM}, which states the other direction ($\kappa$-SEM-VER implies $\kappa$-IND-VER).

\begin{proof}[of Lemma~\ref{thm:IND-SEM}]
	Suppose that a scheme $S$ is $\kappa$-SEM-VER, and let $\advA = (\advA_1, \advA_2, \advA_3)$ be an arbitrary QPT adversary for the IND-VER indistinguishability game for this scheme. By defintion of $\kappa$-SEM-VER, for $\advA_2$ there exists $\simS$ such that for all QPTs $\inM$ and $\disD$, the equation from Definition~\ref{def:SEM-VER} holds with $\advA := \advA_2$. We choose $\inM$ and $\disD$ as in the figure below. More precisely, $\inM$ does: (i.) run $\advA_1$ on its input (ii.) prepare the state $\ketbra{0^n}{0^n}$, plus a random bit $r \in_R \{0,1\}$, and store them in the side information register $R_2$, and (iii.) swap the quantum states in $X$ and $R_2$ conditioned on $r$. We also choose $\disD$ to (i.) run $\advA_3$ on the appropriate input wires, (ii.) either apply $\Phi_c$ or $\remove$ on the quantum state in the register $R_2$, conditioned on the accept/reject flag, (iii.) swap those wires back (again, conditioned on $r$), and finally (iv.) output 1 if $\advA_3$'s output was correct (i.e.\ equal to $r$), and 0 otherwise.

	\begin{center}
	\scalebox{0.75}{
	\makebox[\textwidth][c]
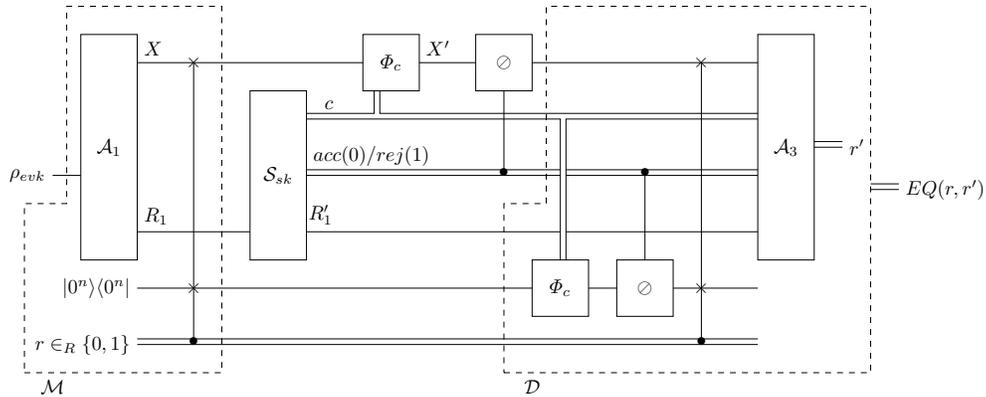
{
		\begin{tikzpicture}
		\draw (-0.5,1.5) -- (0,1.5);
		\node[anchor=east] at (-0.5,1.5) {$\rho_{evk}$};
		\draw (0,0) rectangle (1,4);
		\node at (0.5,2) {$\advA_1$};
		
		\draw (1,0.5) -- (3,0.5);
		\node[anchor=south west] at (1,0.5) {$R_1$};
		\draw (1,-0.5) -- (12,-0.5);
		\node[anchor=east] at (1,-0.5) {$\ketbra{0^n}{0^n}$};
		\draw (1,3.5) -- (5,3.5);
		\node[anchor=south west] at (1,3.5) {$X$};
		\draw (1,-1.5) -- (12,-1.5);
		\draw (1,-1.4) -- (12,-1.4);
		\node[anchor=east] at (1,-1.5) {$r \in_R \{0,1\}$};
		
		\node at (2,-1.45) {$\bullet$};
		\draw (2,-1.4) -- (2,3.5);
		\node at (2,3.5) {$\times$};
		\node at (2,-0.5) {$\times$};
		
		\draw (3,0) rectangle (4,3);
		\node at (3.5,1.5) {$\simS_{sk}$};
		
		\draw (4,2.5) -- (8.5,2.5) -- (8.5,0);
		\draw (4,2.6) -- (5.2,2.6) -- (5.2,3);
		\draw (5.3,3) -- (5.3,2.6) -- (12,2.6);
		\draw (8.6,0) -- (8.6,2.5) -- (12,2.5);
		\node[anchor=south] at (4.4,2.55) {$c$};
		\draw (4,1.6) -- (8.5,1.6);
		\draw (4,1.5) -- (8.5,1.5);
		\draw (8.6,1.6) -- (12,1.6);
		\draw (8.6,1.5) -- (12,1.5);
		\node[anchor=south west] at (4,1.55) {$acc(0)/rej(1)$};
		\draw (4,0.5) -- (8.5,0.5);
		\draw (8.6,0.5) -- (12,0.5);
		\node[anchor=south] at (4.25,0.5) {$R_1'$};
		
		\draw (5,3) rectangle (6,4);
		\node at (5.5,3.5) {$\Phi_c$};
		\draw (6,3.5) -- (7,3.5);
		\node[anchor=south west] at (6,3.5) {$X'$};
		
		\node at (7.5,1.55) {$\bullet$};
		\draw (7.5,1.6) -- (7.5,3);
		\draw (7,3) rectangle (8,4);
		\node at (7.5,3.5) {\faBan};
		\draw (8,3.5) -- (12,3.5);
		
		\filldraw[fill=white] (8,-1) rectangle (9,0);
		\node at (8.5,-0.5) {$\Phi_c$};
		
		\node at (10,1.55) {$\bullet$};
		\draw (10,1.6) -- (10,0);
		\filldraw[fill=white] (9.5,-1) rectangle (10.5,0);
		\node at (10,-0.5) {\faBan};
		
		\node at (11,-1.45) {$\bullet$};
		\draw (11,-1.4) -- (11,3.5);
		\node at (11,-0.5) {$\times$};
		\node at (11,3.5) {$\times$};
		
		\draw (12,0) rectangle (13,4);
		\node at (12.5,2) {$\advA_3$};
		\draw (13,2) -- (13.5,2);
		\draw (13,2.1) -- (13.5,2.1);
		\node[anchor=west] at (13.5,2) {$r'$};

		\draw[dashed] (-1,-2) -- (2.5,-2) -- (2.5,4.5) -- (-0.25,4.5) -- (-0.25,1) -- (-1,1) -- (-1,-2);
		\node[anchor=north] at (-0.5,-2) {$\inM$};
		
		\draw[dashed] (7.5,-2) -- (14,-2) -- (14,4.5) -- (8.25,4.5) -- (8.25,1) -- (7.5,1) -- (7.5,-2);
		\node[anchor=north] at (8,-2) {$\disD$};
		\draw (14,1.25) -- (14.5,1.25);
		\draw (14,1.35) -- (14.5,1.35);
		\node[anchor=west] at (14.5,1.25) {$EQ(r,r')$};
		\end{tikzpicture}
	}}
	\end{center}
Note that these choices ensure that the real channel is an execution of the IND-VER game. In the ideal scenario, $\advA_3$ receives \emph{exactly} the same state in the cases $r=0$ and $r=1$. Hence, the best he can do is guess, and the probability that $r' = r$ (and thus that $\disD$ outputs 1)is at most $\frac{1}{2}$.
		
		By the assumption that $S$ is $\kappa$-SEM-VER, the probability that $\disD$ outputs 1 in the real scenario can only be negligibly higher than in the ideal case. As discussed above, the real scenario corresponds exactly to the adversary $\advA$ playing the IND-VER game. Therefore, the winning probability for $\advA$ (i.e.\ the probability that $\VerGame{S} = 1$) is at most negligibly (in $\kappa$) higher than $\frac{1}{2}$. 
		\qed
\end{proof}

\section{Security of verification in $\TC$.}\label{app:TC-proof}

The trap code is proven secure in its application to one-time programs~\cite{BGS13}. Broadbent and Wainewright proved authentication security (with an explicit, efficient simulator)~\cite{BW16}. One can use similar strategies to~\cite{BW16, BGS13} to prove $\kappa$-IND-VER for $\TC$. 

\begin{theorem}\label{thm:tc-security-1-supp}For any adversary $\advA$,
\[
\Pr[\VerGame{\TC} = 1] \leq \frac{1}{2} + \negl(\kappa) \,,
\]
and thus $\TC$ is a $\kappa$-IND-VER secure (somewhat) homomorphic encryption scheme.
\end{theorem}

The proof (again following, e.g., \cite{BW16}) will use the following lemma, the \emph{Pauli Twirl}~\cite{DCEL09}.
\begin{lemma}[Pauli Twirl]\label{lem:pauli-twirl}
	Let $\rho$ be an arbitrary $n$-qubit state. Then for any Pauli operators $P$, $P'$ it holds that
	\[
	\frac{1}{4^n} \sum_{a,b \in \{0,1\}^n} (\X^a \Z^b)^\dagger P \X^a \Z^b \rho 
	\X^a \Z^b P'^\dagger (\X^a \Z^b)^\dagger = \begin{cases}
			P \rho P^\dagger & \text{if}\ P = P'  \\
			0 & \text{otherwise} 
	\end{cases}
	\]
\end{lemma}

\begin{proof}[Theorem \ref{thm:tc-security-1-supp}]
Let $\advA = (\advA_1, \advA_2, \advA_3)$ be an adversary for $\TC$, for the $\VerGame{S}$ security game. Let $sk = (\pi, x, z)$ be the uniformly random keys, with
$\pi \in S_{3m}$, $x,z \in \{0,1\}^{3m n}$. Let $\CSS$ be a $[[m,1,d]]$ CSS code,
that can correct up to $d_c$ errors. We can let $d_c$ be the security parameter $\kappa$, and then 
$d = 2d_c +1$ and $m$ depends on the exact properties of $\CSS$.

First note that for $\TC$, the circuit $c$ which is output by $\advA_2$ cannot in any way depend on the bit $r$: All qubits output by $\TC.\Enc$ are encoded with the quantum one-time pad, and therefore will look completely mixed whether or not the real or dummy input is given to $\advA_2$.
Also, $c$ is measured when it is supplied to $\VerDec$ as classical information. Therefore, we can in general view $\advA$ as a probabilistic mixture of adversaries for different choices of $c$. From now on, we assume that $\advA$ uses an arbitrary fixed $c$ without loss of generality (since it can always use the circuit $c$ that wins the game with highest probability).

Next, observe that the accept probability of $\VerDec$ within the $\VerGame{S}$ game is independent of the random choice $r$. The decryption procedure only looks at the trap qubits when choosing whether to accept or reject, and so we can imagine delaying undoing the quantum one-time pad on the data qubits until after the accept or reject choice -- which cannot depend on $r$ at all since the encrypted data always looks completely mixed.

In the reject case, $\VerDec$ outputs a fixed quantum state, and the quantum one-time pad that is applied to the input of $\TC.\Enc$ will never be revealed. So in that case $\advA_3$ will never be able to do better than a random guess.
To prove security, it then suffices to argue that the state $\advA_2$ outputs in the $r=0$ case is close to the state $\advA_2$ outputs in the $r=1$ case, conditioned on $\VerDec$ accepting.

Now let $D$ be the quantum operations that are performed by the honest evaluation circuit, i.e., the list of $\CNOT$ gates applied transversally to the encrypted qubits. Let $B \in \{I, \text{comp}, H\}^n$ describe for each logical qubit whether it is unmeasured, measured in the computational basis, or measured in the Hadamard basis respectively.
To simplify notation, we assume without loss of generality that $\mathcal H_R = \mathcal H_{R'}$ and that
$\advA_2$ can be written as a unitary operation that consists of first applying an arbitrary $U$ acting on $\mathcal H_R \otimes \mathcal H_C$, and then applying the honest actions $D$.

Define $M_{I} = I^{\otimes m} \otimes \ketbra{0}{0}^{\otimes m} \otimes \ketbra{+}{+}^{\otimes m}$ as the projector corresponding to accepting the traps of an unmeasured qubit (after undoing the permutation and quantum one-time pad). Similarly define $M_{\text{comp}} = I^{\otimes m} \otimes \ketbra{0}{0}^{\otimes m} \otimes I^{\otimes m}$ and $M_{H} = I^{\otimes m} \otimes I^{\otimes m} \otimes \ketbra{+}{+}^{\otimes m}$ as the projectors corresponding to accepting the traps of a measured qubit.

The $\VerDecMeasurement$ function does not undo the $\Z$ part of the quantum one-time pad in case
of computational-basis measurement (or the $\X$ corrections for the Hadamard basis), making it not possible to immediately apply the Pauli twirl (Lemma~\ref{lem:pauli-twirl}) as with the unmeasured qubits.
Consider a scheme where the decryption procedure for the measurement
would first undo the entire quantum one-time pad, and only measure the qubits afterward:
this scheme would be functionally completely equivalent to the actual $\TC$ scheme.
(In~\cite{BGS13} this property is called the equivalence between decode-then-measure and measure-then-decode.)
Also, note that we can write the Pauli operators and the conditional Paulis as occuring after all
other gates; because of the commutation rules between the Clifford group and the Pauli group, they
will just correspond to different Pauli operations applied later.
Since the Pauli key updates only occur on the data qubits and
do not change the acceptance probability at all,
we can rewrite the decryption procedure as first checking the traps using the keys without the Paulis, and only then apply the (conditional) Paulis after the checks.

Let $\sigma \in \mathcal{H}_{RX}$ be the output state of $\advA_1$. Write $\Lambda(\cdot) = \CSS.\Encode(\cdot) \otimes \ketbra{0}{0}^{\otimes m} \otimes \ketbra{+}{+}^{\otimes m}$ as the channel representing the part of the encryption which encodes the qubits and appends the traps and define $\rho= \Bigl(\mathrm{id}_R \otimes \bigotimes^n_{i = 1} \Lambda \Bigr) (\sigma)$. Now the state $(\id_R \otimes \X^x \Z^z \pi^{\otimes n}) \rho (\id_R \otimes \X^x \Z^z \pi^{\otimes n})^{\dagger}$ is the encrypted input to $\advA_2$. The plaintext that $\VerDec$ holds right before
running $\CSS.\Decode$, projected to the accepting case, equals
\begin{equation*}
\bigotimes_{i\in [n]} M_{B_i} \Ex_{\pi \in S_{3m}} \Bigl[ \Ex_{x,z \in \{0,1\}^{3mn}} [
\pi^{\dagger \otimes n} \X^{f_c(x)} \Z^{f_c(z)} D U  \X^{x} \Z^{z} \pi^{\otimes n}  \rho
\pi^{\dagger \otimes n} \X^{x} Z^{z}  U^{\dagger} D^{\dagger} \X^{f_c(x)} \Z^{f_c(z)} \pi^{\otimes n} ] \Bigr]
\end{equation*}
where $U$ acts on the reference system $R$ and the ciphertext register, and all
other operations only act on the ciphertext register.\footnote{The expectation value is always taken over the uniform distribution, e.g., $\Ex_{\pi \in S_{3m}}$ is nothing more than a short way of writing $\frac{1}{(3m)!} \sum_{\pi \in S_{3m}}$.}
The transformation $f_c$ represents the updating of the quantum one-time pad keys as function of the applied circuit -- these are the keys that are used by the decryption circuit.

Now define $D'$ as the unitary operation which applies $\CNOT$ gates transversally on the $3m + 3m$ qubits when listed in $c$. By construction of the key-update rules, we have that $\X^{f_c(x)} \Z^{f_c(z)} \pi^{\dagger \otimes n } D = D' \X^{x} \Z^{z} \pi^{\dagger \otimes n }$. 
Using that identity and the Pauli twirl (Lemma~\ref{lem:pauli-twirl}) we decompose $U$ into a probabilistic mixture of Pauli operations:
\begin{align*}
\bigotimes_{i\in [n]} M_{B_i} D' \Ex_{\pi \in S_{3m}} \Bigl[ \Ex_{x,z \in \{0,1\}^{3mn}} [
\pi^{\dagger \otimes n}  \X^{x} \Z^{z} U \X^{x} \Z^{z} \pi^{\otimes n}  \rho  
\pi^{\dagger \otimes n} \X^{x} \Z^{z} U^{\dagger} \X^{x} \Z^{z} \pi^{\otimes n} & ] \Bigr] D'^{\dagger} =\\
\bigotimes_{i\in [n]} M_{B_i} D' \Ex_{\pi \in S_{3m}} \Bigl[ \Ex_{P \in {\mathcal{P}}^{3mn}} [ |\alpha_P|^2 \pi^{\dagger \otimes n}  (P\otimes U_P) \pi^{\otimes n} \rho 
\pi^{\dagger \otimes n} (P\otimes U^{\dagger}_P)  \pi^{\otimes n} ] \Bigr] D'^{\dagger}  &
\end{align*}

Expressions of this form were carefully analyzed in the earlier trap-code security proofs --
we will for completeness finish our security sketch, but see, e.g., \cite{BW16} for a more precise analysis.

First observe that if these Pauli operators do not change any logical qubit, this expression will be exactly the same as the state that $\advA_3$ receives in the $r=1$ case -- namely the claimed circuit $c$, as represented by $D'$, the measurements, and the (conditional) Paulis that will be effectively performed by decryption, applied to the data.
Consider what form the Pauli operator $P$ would need to change a specific logical qubit $i$.
First consider an unmeasured qubit $i$. 
Because $\CSS$ can correct up to $d_c$ errors, only those Paulis that are non-identity on more
than $d_c$ qubits will cause the logical qubit after decoding to change.
Say without loss of generality that this Pauli
operator $P_i$ has an $\X$ on at least $d_c/2$ out of the $3m$ physical qubits that encode $i$.
(If the operator consists of more $\Z$ components than $\X$ components, we could argue using $\Z$ instead.)
Consider the probability for a randomly chosen $\pi \in S_{3m}$ that all these $\X$ do not end up in the positions $m+1$ to $2m$, i.e., each misses the computational basis traps.
For each $\X$, the probability of missing all trap positions, conditioned on no trap being hit yet, is always at most $2/3$. Therefore the probability that all traps are missed is at most $(2/3)^{d_c/2}$.
A more careful combinatorial analysis which includes the $\Z$ flips improves this to $(2/3)^{d_c}$~\cite{BGS13,BW16}, but this simple bound suffices for us.

Now, consider the case that $i$ is a qubit on which a computational-basis measurement has been performed,
of which only the corresponding traps are checked. 
For these qubits, the Pauli $\Z$ parts of the attack are not detected, but they also do not change the output: Since the data qubits are measured, only the $\X$ Paulis will change anything in the data.
Therefore, the operator $P_i$ will have to contain at least $d_c$ $\X$ Paulis on the $3m$ physical qubits.
Now repeating the same argument as for the unmeasured qubits, we see that the probability over a random permutation $\pi$ that all traps are missed is at most $(2/3)^{d_c}$.
The analogous argument works for the Hadamard-basis measurements.

To conclude, the part of the output of $\advA_2$ that has been changed from that what would come out of the honest evaluator, and still is accepted, has norm at most $(2/3)^{d_c}$, both in case $r=0$ and $r=1$.
This norm gives an upper bound to the trace distance between the states that $\advA_3$ receives in the
$r=0$ case and the $r=1$ case, since for all lower-weight Pauli attacks these states are exactly the same (by the error-correction property of $\CSS$). The final guessing probability is then bounded as
\[
\Pr[\VerGame{\TC} = 1] \leq \frac{1}{2} + \frac{1}{2} \left(\frac{2}{3} \right)^{d_c} \,.
\]
Since we picked the parameters of $\CSS$ such that $d_c$ scaled with $\kappa$, this completes the proof.
\qed
\end{proof}

\section{Security of $\TC$ with multiple encryptions}\label{app:TC-defs}

In Section~\ref{sec:ttp-verifiability}, we will use the IND-VER property of $\TC$ to prove verifiability for our new scheme. In order to achieve this, we will actually need a slightly stronger notion of verifiability for $\TC$: IND-VER-$n$, where the adversary is allowed to submit plaintexts in multiple rounds, which are either all encrypted or all swapped out. In this subsection, we show that $\TC$ also fulfills this stronger notion. For our purposes in Section~\ref{sec:ttp-verifiability}, it suffices to show that $\TC$ is secure against an adversary that is allowed two rounds (IND-VER-2), but the definitions and proof trivially extend to the general case.

\begin{definition}[VER-2 indistinguishability game $\VerTwoGame{S}$]\label{def:VER-2-game}
	For an adversary $\advA = (\advA_0, \advA_1, \advA_2, \advA_3)$, a scheme $S$, and a security parameter $\kappa$, $\VerTwoGame{S}$ is the following game:
	
	
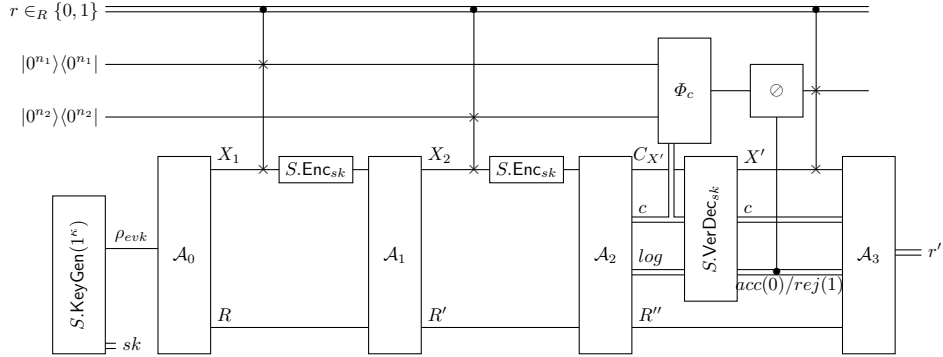
\begin{center}
		\begin{tikzpicture}[scale=0.7, every node/.style={scale=0.7}]
		\draw (-4,0) rectangle (-3,3);
		\node[rotate=90] at (-3.5,1.5) {$S.\KeyGen(1^{\kappa})$};
		\draw (-3,2) -- (-2,2);
		\draw (-3,0.1) -- (-2.8,0.1);
		\draw (-3,0.2) -- (-2.8,0.2);
		\node[anchor=south] at (-2.5,2) {$\rho_{evk}$};
		\node[anchor=west] at (-2.8,0.2) {$sk$};
		
		\node[anchor=east] at (-3,4.5) {$|0^{n_2}\rangle\langle0^{n_2}|$};
		\node[anchor=east] at (-3,5.5) {$|0^{n_1}\rangle\langle0^{n_1}|$};
		\node[anchor=east] at (-3,6.5) {$r \in_R \{0,1\}$};
		\draw (-3,6.5) -- (11.5,6.5);
		\draw (-3,6.6) -- (11.5,6.6);
		\draw (-3,5.5) -- (7.5,5.5);
		\draw (-3,4.5) -- (7.5,4.5);
		
		\node at (-1.5,1.875) {$\advA_0$};
		\draw (-2,0) rectangle (-1,3.75);
		\draw (-1,0.5) -- (2,0.5);
		\draw (-1,3.5) -- (2,3.5);
		\node[anchor=south west] at (-1,0.5) {$R$};
		\node[anchor=south west] at (-1,3.5) {$X_1$};
		
		\node at (2.5,1.875) {$\advA_1$};
		\draw (2,0) rectangle (3,3.75);
		\draw (3,0.5) -- (6,0.5);
		\draw (3,3.5) -- (6,3.5);
		\node[anchor=south west] at (3,0.5) {$R'$};
		\node[anchor=south west] at (3,3.5) {$X_2$};
		
		\filldraw[fill=white] (0.3,3.25) rectangle (1.7,3.75);
		\node at (1,3.5) {$S.\Enc_{sk}$};
		\node at (0,6.53) {$\bullet$};
		\draw (0,6.5) -- (0,3.5);
		\node at (0,5.5) {$\times$};
		\node at (0,3.5) {$\times$};

		\filldraw[fill=white] (4.3,3.25) rectangle (5.7,3.75);
		\node at (5,3.5) {$S.\Enc_{sk}$};
		\node at (4,6.53) {$\bullet$};
		\draw (4,6.5) -- (4,3.5);
		\node at (4,4.5) {$\times$};
		\node at (4,3.5) {$\times$};
		
		\filldraw[fill=white] (6,0) rectangle (7,3.75);
		\node at (6.5,1.875) {$\advA_2$};
		\draw (7,0.5) -- (11,0.5);
		\draw (7,1.5) -- (8,1.5);
		\draw (7,1.6) -- (8,1.6);
		\draw (7,2.5) -- (8,2.5);
		\draw (7,2.6) -- (7.7,2.6) -- (7.7,4);
		\draw (7.8,4) -- (7.8,2.6) -- (8,2.6);
		\draw (7,3.5) -- (7.7,3.5);
		\draw (7.8,3.5) -- (8,3.5);
		\node[anchor=south west] at (6.9,3.5) {$C_{X'}$};
		\node[anchor=south west] at (7,2.55) {$c$};
		\node[anchor=south west] at (7,1.55) {$log$};
		\node[anchor=south west] at (7,0.5) {$R''$};
		
		\filldraw[fill=white] (7.5,4) rectangle (8.5,6);
		\node at (8,5) {$\Phi_c$};
		\draw (8.5,5) -- (11.5,5);

		\draw (8,1) rectangle (9,3.75);
		\node[rotate=90] at (8.5,2.375) {$S.\VerDec_{sk}$};
		\draw (9,1.5) -- (11,1.5);
		\draw (9,1.6) -- (11,1.6);
		\draw (9,2.5) -- (11,2.5);
		\draw (9,2.6) -- (11,2.6);
		\draw (9,3.5) -- (11,3.5);
		\node[anchor=north] at (10,1.6) {$acc(0)/rej(1)$};
		\node[anchor=south west] at (9,2.55) {$c$};
		\node[anchor=south west] at (9,3.5) {$X'$};
		
		\node at (9.75,1.55) {$\bullet$};
		\draw (9.75,1.6) -- (9.75,4.5);
		\filldraw[fill=white] (9.25,4.5) rectangle (10.25,5.5);
		\node at (9.75,5) {\faBan};
		\node at (10.5,6.53) {$\bullet$};
		\draw (10.5,6.5) -- (10.5,3.5);
		\node at (10.5,5) {$\times$};
		\node at (10.5,3.5) {$\times$};
		
		\draw (11,0) rectangle (12,3.75);
		\node at (11.5,1.875) {$\advA_3$};
		\draw (12,1.875) -- (12.5,1.875);
		\draw (12,1.975) -- (12.5,1.975);
		\node[anchor=west] at (12.5,1.975) {$r'$};
		\end{tikzpicture}
	\end{center}
	Here, $n_1$ and $n_2$ are the respective dimensions of the $X_1$ and $X_2$ registers.
\end{definition}

\begin{definition}[$\kappa$-IND-VER-2]
	A vQFHE scheme $S = (\KeyGen, \Enc, \Eval, \VerDec)$ has \emph{2-round}\\ \emph{$\kappa$-indistinguishable verification} if for any QPT adversary $\advA = (\advA_0,\advA_1, \advA_2, \advA_3)$,
	\[
	\Pr[\VerTwoGame{S} = 1] \leq \frac{1}{2} + \negl(\kappa).
	\]
	Here, the probability is taken over $\KeyGen(1^{\kappa}), \Enc, \VerDec$, and $\advA$.
\end{definition}

\begin{lemma}\label{thm:tc-security-2-supp}
	$\TC$ is $\kappa$-IND-VER-2.
\end{lemma}

\begin{proof}
	Let $\advA = (\advA_0, \advA_1, \advA_2, \advA_3)$ be an arbitrary polynomial-time adversary for the VER-2 indistinguishability game for $\TC$. For notational convenience, write the secret key as $sk = (\pi, x_1, z_1, x_2, z_2)$, where $x_1$ and $z_1$ are lists of $3mn_1$ bits, sufficient for encrypting $X_1$, and analogously $x_2$ and $z_2$ are lists of $3mn_2$ bits.
	
	We now slightly alter the VER-2 game in the following way. In the first encryption step of the game, instead of providing $\advA_1$ with $\TC.\Enc_{(\pi, x_1, z_1)}$ applied to the register $X_1$, we provide $\advA_1$ with the halves of $n_1$ EPR pairs, and perform Bell measurements between the other halves and the qubits in $X_1$, after they have been CSS-encoded and permuted with traps. Let the outcomes of these measurements be given by $a,b \in \{0,1\}^{3mn_1}$: $a$ and $b$ describe the effective $\X$ and $\Z$ Paulis that are applied to $X_1$ by these teleportation measurements. To undo these Paulis, we update $sk$ to $(\pi, x_1 \oplus a, z_1 \oplus b, x_2, z_2)$ at this point. Here, $\oplus$ is bitwise addition modulo 2. Since the quantum one-time pad keys $x_1$ and $z_1$ are chosen uniformly at random, and are completely hidden from the perspective of the adversary, the new keys $x_1 \oplus a$ and $z_1 \oplus b$ are valid keys that are sampled from the same distribution.
	Hence, the winning probability of $\advA$ is not affected by this change of the game.
	
	A second small change to the game is the following: instead of performing the Bell measurements and the secret-key update immediately, it is done only after $\advA_1$ has provided its query in $X_2$. Since these actions happen only on wires which are not accessible to $\advA_1$ and otherwise also not touched in this stage of the game, this change also does not affect the execution or outcome of the game in any way.
	
	We have now arrived at an interesting situation: $\advA_1$ only receives halves of EPR pairs, and so its choice for $X_2$ or $R'$ is not based on the first ciphertext received from the challenger -- that ciphertext will only be generated after execution of $\advA_1$. We can merge $\advA_0$ and $\advA_1$ into a single QPT algorithm that produces $X_1$ and $X_2$ simultaneously. When viewed as such, $\advA$ is an adversary for the single-query VER indistinguishability game, and we can conclude that
	\[
	\Pr[\VerTwoGame{\TC} = 1] = \Pr[\VerGame{\TC} = 1].
	\]
	Since we know that the latter probability is bounded by $\frac{1}{2} + \negl(\kappa)$ from Theorem~\ref{thm:tc-security-2}, so is the first.
	\qed
\end{proof}

\end{document}